\documentclass[11pt,reqno]{amsart}
 \usepackage[foot]{amsaddr}
\usepackage[dvipsnames]{xcolor}
\usepackage[utf8]{inputenc}  
\usepackage[T1]{fontenc}    
\usepackage[english]{babel}   
\usepackage{graphicx,amssymb,mathrsfs,amsmath,amsfonts,appendix, subfigure, amsthm, paralist, enumitem, comment, mathtools, esint, derivative, colonequals,bbm}
\usepackage[colorlinks=true, pdfstartview=FitV, linkcolor=blue, citecolor=blue, urlcolor=blue]{hyperref}
\usepackage{comment}
\usepackage{orcidlink}
\usepackage{siunitx}
\usepackage{pgfplots}
\usepackage{subcaption}

\usepackage[normalem]{ulem}

\usepackage{float}
\usepackage{enumitem}
\pgfplotsset{compat=1.18}

\tikzset{
LM/.style = {very thin,
        {Bar[]Stealth}-%
        {Stealth[]Bar}
            },
every edge quotes/.append style= {font=\footnotesize},
  every node/.append style ={font=\small} 
  }

\numberwithin{equation}{section}

\newcounter{cc}[section]
\counterwithin*{cc}{section}

\newtheorem{theorem}[cc]{Theorem}
\newtheorem{lemma}[cc]{Lemma}
\newtheorem{proposition}[cc]{Proposition}
\theoremstyle{definition}
\newtheorem{definition}[cc]{Definition}
\newtheorem{remark}[cc]{Remark}
\newtheorem{assumption}[cc]{Assumption}

\newtheorem{algorithm}[cc]{Algorithm}
\newtheorem{corollary}[cc]{Corollary}

\newcommand{\R}{\mathbb{R}}

\newcommand{\N}{{\mathbb N}}
\newcommand{\E}{{\mathbb E}}

\newcommand{\rme}{{\rm e}}
\newcommand{\rmd}{{\rm d}}

\newcommand{\reduc}{C_{i}^{\mathrm{r}}}

\newcommand{\finaltr}{\mathfrak{T}^{\mathrm{r}}}

\newcommand{\cf}[1]{{\mathbbm 1}_{\{#1\}}}

\newcommand{\vc}[1]{\boldsymbol{#1}}

\renewcommand{\cf}[1]{{\mathbbm 1}_{\{#1\}}}

\DeclareMathOperator*{\argmin}{arg\,min}

\newcommand{\fix}[1]{#1^\bullet}
\newcommand{\free}[1]{#1^\circ}

\newcommand{\inI}[1]{\mathfrak{I}_{#1}}

\newcommand{\abs}[1]{\left| #1 \right|}
\newcommand{\norm}[1]{\Vert #1\Vert}

\title{Pointwise decay of cumulants and scaling regimes of chaotic states}

\author{Stefano Rossi\textsuperscript{*}}
\address{\textsuperscript{*}Dipartimento di Matematica Guido Castelnuovo, Sapienza Università di Roma, Piazzale Aldo Moro 5, 00185 Rome, Italy}
\email{stefano.rossi2@uniroma1.it}

\author{Sergio Simonella\textsuperscript{*}}
\email{sergio.simonella@uniroma1.it}

\author{Aleksis Vuoksenmaa\textsuperscript{*}}
\email{aleksis.vuoksenmaa@uniroma1.it}

\begin{document}
    \begin{abstract}
    We consider nonequilibrium chaotic states of classical interacting particles in a family of scaling limits motivated by kinetic theory. The particle activity and the strength of the pair interaction are rescaled simultaneously,  allowing for both dilute and non-dilute regimes. Using cluster expansion, we derive pointwise bounds for truncated correlation functions (cumulants) of arbitrary order. We cover stable finite-range interactions and repulsive Riesz potentials, including non-integrable tails in bounded macroscopic domains.
    \end{abstract}
    \maketitle
    \section{Introduction}

The passage from a microscopic description of matter 
in terms of interacting particles
to macroscopic laws 
is a classical problem of statistical mechanics. 
Both equilibrium thermodynamics and effective evolution equations emerge through limiting procedures, 
whose outcome depends on
how particle number, interaction range, and coupling strength are scaled. 
There is no unique scaling limit: identifying the relevant regimes, 
and understanding which are compatible with a given interaction potential, is part of the mathematical problem itself.

In kinetic theory, 
the Boltzmann--Grad limit provides an instructive example. 
Recent advances in its rigorous analysis have exploited
refined mathematical formulations of the well-known molecular chaos property.
A strong quantitative control of the correlations 
generated by microscopic dynamics has been proved crucial 
in the derivation of effective equations 
\cite{DengHaniMa2024,BGSS26, PulviSimo26,Pulvirenti2021}. 
Cumulants, or truncated correlation functions, are the technical tool used to
isolate the connected parts of multi-particle correlations, 
and effectively quantify the departure from statistical independence. 

In statistical mechanics, cumulants are estimated through classical cluster expansion methods. 
Such methods have been suitably adapted to study genuine dynamical problems \cite{PulviSimo17,DengHaniMa2024}.
However, rigorous results on nonlinear kinetic dynamics 
still mostly concern dilute regimes and hard-sphere type scattering. 
Different scalings and collision mechanisms are substantially less understood, as well as the inclusion of collective effects. 
Equilibrium statistical mechanics, on the other hand, provides a rich theory of correlations,
ranging from low-activity and high-temperature expansions, to long-range and Kac-type interactions.

In this paper, 
we revisit the static problem, with the goal of providing a unified picture across different scaling regimes. 
We investigate in detail how the spatial decay of cumulants
and the dependence on their order are affected by pair potentials with finite or infinite range, 
and by the simultaneous rescaling of particle activity and coupling strength.

More precisely, we consider particle systems with interaction localized in space on a typical length of order \(\varepsilon>0\)
and with a coupling strength \(\varepsilon^\alpha\), with \(\alpha \geq 0\).
The microscopic Hamiltonian reads as
\begin{align*}
    \label{eq:N-hamiltonian}
         \sum_{i=1}^{n}\frac{v_i^2}{2} + \sum_{i,j=1}^{n} \cf{i\neq j} \varepsilon^\alpha \varphi\left(\frac{x_i-x_j}{\varepsilon}\right), \quad  (x_i,v_i) \in  \Lambda \times \R^d,
    \end{align*}
    where \(\Lambda \subset \R^d\) is the spatial domain and \(\varphi\) is the pairwise interaction potential of the system, which is assumed to be central and to be finite (compactly supported) or infinite range.
    Given \(\beta>0\) and \(\delta>-d\), we consider a grand canonical ensemble of identical particles, described by a sequence of exchangeable 
    densities \(W^{\varepsilon}=\{W^\varepsilon_n\}_{n \geq 0}\), with 
    \[
    W^{\varepsilon}_n(\vc{z}_n) \propto   \mu_{\varepsilon,\delta}^nf^{\otimes n}(\vc{z}_n)\exp\Biggl\{-\beta\,\sum_{i \neq j}\varepsilon^\alpha \varphi\left(\frac{x_i-x_j}{\varepsilon}\right)\Biggr\} , \quad \vc{z}_n= (x_i,v_i)_{i=1}^{n} \in \Omega_n = (\Lambda \times \R^d)^{n},
    \]
     where
    \(f\) is an assigned non-equilibrium one-particle density and \(\mu_{\varepsilon,\delta}=\varepsilon^{-d-\delta}\) is the activity parameter of the ensemble.
    The scaling parameter \(\varepsilon\) can be regarded as the ratio 
    between the characteristic microscopic and macroscopic units of the system.
    The behavior of the system in the limit \(\varepsilon\to 0\)
    depends crucially on the scaling considered, 
    and therefore on the parameters \(\alpha\) and \(\delta\)
    that influence, respectively, the strength of the pairwise interaction
    and the typical number of particles in the system.
    Different choices of the scaling and their relevance to kinetic theory are recalled in 
    Subsection \ref{subsection:comments-on-the-scaling-regimes} below.
   
    In a range of admissible scaling regimes
    depending on \(\alpha\) and \(\delta\), 
    the grand canonical measure \(W^\varepsilon\) can be regarded as
    a quasi-Poisson distribution; or ``\emph{maximally factorized} (or \emph{chaotic}) \emph{state}''. 
    Such states constitute a natural, simple initial data in the derivation of collisional kinetic equations; see e.g.~\cite{ehlers_time_1975, CercignaniIllnerPulvirenti1994, spohn_large_1991}.
    They can be also viewed as inhomogeneous Gibbs-type distributions,
    obtained by weighting a Poisson reference measure with the Boltzmann factor of the pair interaction. 
    For an isothermal Maxwellian velocity distribution, they have the form of a local Gibbs state; 
    more general local Gibbs states with spatially varying temperature could be also included, 
    by appropriately modifying the pair weights, subject to stability assumptions.
    Maximally chaotic states at low-density regime have also been used as an ansatz for the dynamical distribution of particles, giving rise to a family of finite-order corrections to the Boltzmann equation; see \cite{Resibois1978, GerasimenkoGapyak2012}.
      
    As the main result of the present work, 
    we obtain precise pointwise estimates for the general \(j^{\text{th}}\) order cumulants 
    of the maximally chaotic states \(W^\varepsilon\), 
    controlling explicitly both the dependence on the combinatorial graph enumeration factors \(j^{j-2}\) and the decay rates in the interparticle distance with the microscopic parameter: exponential decay for finite-range interactions and polynomial decay for long-range potentials. The structure of these decay rates is uniform across different scaling regimes; 
    see Theorem~\ref{thm:main-theorem}.
    
    Our main technique 
    is the cluster expansion, which allows one to estimate the cumulants of the system through a series of graph integrals. This method has a long history in the study of thermodynamic limits of classical continuous systems and discrete lattice systems; see, e.g.,~\cite{gallavotti_correlation_1968, ruelle1969statistical, gruber_general_1971}. 
    Pointwise decay bounds are also known in the literature as instances of the ``strong cluster/decrease property'', see~\cite{duneau_decrease_1973, DuneauIagolnitzerSouillard1974, duneau_cluster_1976, minlos_estimates_1977}. Cumulants of maximally chaotic states in the Boltzmann--Grad scaling \((\alpha=0, \delta=-1)\) were also studied in~\cite{SimonellaWinter2026}.

    For infinite-range interactions, there is a technical difficulty in integrating out the background vertices of a tree,
    while retaining roots on prescribed particles: this may lead to deteriorated exponents or loss in the combinatorics.  
    To get precise pointwise estimates, preserving both the algebraic decay exponent and the factorial dependence on the order of cumulants, we devise a specific algorithmic procedure. This allows us to avoid a standard reduction to linear trees (chains);  see e.g.\,\cite{duneau_cluster_1976, minlos_estimates_1977} and more recently \cite{dorlas_correlation_2020, SimonellaWinter2026}, where an extra square root appears in the decay estimates.
    In our work, we include infinite--range Riesz potentials of type \(|x|^{-s}\) with \(s\) 
    belonging to a range of values which depends on the choice of \(\alpha\) and \(\delta\) 
    through the balance between decay, activity and coupling strength. 
    For instance in the Boltzmann-Grad regime, the decay bound is valid up to \(s>d-1\) (improving the result in~\cite{SimonellaWinter2026}), and including the endpoint $s=d-1$ under an additional smallness condition.
    
    The rest of this article is organized as follows: 
    in Section~\ref{sec2}, we introduce the class of grand canonical states 
    we want to consider and present the main results. 
    We also discuss the different scaling regimes that are relevant from the point of view of kinetic theory and thermodynamics.
    The pointwise bounds for finite or infinite range potentials are stated in Theorem~\ref{thm:main-theorem}, 
    while the description of the related fluctuation field is given in Corollary \ref{corollary1}, using the \(L^1\) bounds obtained in Proposition~\ref{prop:L1cum}. 
    Section~\ref{sec3} is devoted to the proof of Theorem~\ref{thm:main-theorem}, Proposition~\ref{prop:L1cum}, and Corollary~\ref{corollary1}. 
    Being slightly more technical and of independent interest, we postpone to Section~\ref{sec4} the proof of the main combinatorial estimate of the tree integrals in the case of long--range potentials.

    \section{Set-up and main results}
    \label{sec2}

    In what follows, 
    we consider a system of \(\mathcal{N}\) particles of unit mass in a domain of the Euclidean space \(\Lambda \subset \R^d\).
    We work in the grand canonical ensemble, where the number of particles \(\mathcal{N}\) itself is a random variable.
    Particles
    interact through a pairwise, translation-invariant potential $\varphi :\R^d \to \overline \R \coloneqq \R \cup \{+\infty\}$.
    The potential is assumed to belong to one of the following classes:
    \begin{enumerate}
    \item[(i)] \emph{Finite range} potentials: potentials $\varphi$ satisfying, without loss in generality, 
    \begin{equation}
    \label{eq:support}
    \text{supp}(\varphi) \subseteq B_1(0)\coloneqq\{x \in \R^d\, |\, |x|\leq 1\}.
    \end{equation}
    \item[(ii)] \emph{Riesz type} potentials: potentials $\varphi$ of the form
    \begin{align}
    \label{eq:Riesz}
        \varphi(x) = \frac{1}{\abs{x}^s}, \quad x \in \R^d \setminus \{0\}, \quad s>0.
    \end{align}
    \end{enumerate}
 
    For finite range potentials, we assume the following standard stability condition.
    \begin{assumption}[Stability]
    The pair potential \(\varphi \colon \R^d  \to \overline{\R}\) is assumed to be stable, i.e. we assume that
    there exists \(B \geq 0\) such that, for every \(n \in \N\) and every configuration \((x_1,\dots, x_n) \in (\R^d)^n\)
    \begin{align}
    \label{eq:stability}
        \sum_{1\leq i<j \leq n} \varphi(x_i -x_j) \geq -B n.
    \end{align}
    \end{assumption}
    Observe that the Riesz potential \(|x|^{-s}\) is non--negative and therefore automatically satisfies the stability condition \eqref{eq:stability} with \(B=0\).

    As already discussed in the introduction, we consider 
    a parameter \(\varepsilon>0\) that denotes the ratio between the characteristic macroscopic and microscopic length units of the system. 
    The macroscopic behavior of the system in the limit \(\varepsilon\to 0\)
    depends on the strength of the interaction
    and on
    the scaling of \(\E_\varepsilon[\mathcal{N}]\), i.e.
    the typical number of particles in the system.
    Here \(\E_\varepsilon\) is the expectation taken with respect to the reference grand canonical distribution of the particles,
    which itself depends on the activity parameter \(\mu_\varepsilon\) of the grand canonical ensemble.
    The framework considered in this work describes various scaling regimes through the following general scaling assumption.

    \begin{assumption}[Scaling regimes]
        Let \(\alpha \geq 0\) and \(\delta >-d\). We assume that the activity parameter \(\mu_{\varepsilon, \delta}\) of the grand canonical measure scales as
        \begin{align}
        \label{eq:scalingreg}
            \mu_{\varepsilon, \delta} = \varepsilon^{-(d +\delta)}, \quad \varepsilon > 0,
        \end{align}
        while the pair potential \(\varphi\) is rescaled according to
        \begin{align*}
            \varphi_{\varepsilon,\alpha}(x) = \varepsilon^{\alpha}\varphi\left(\frac{x}{\varepsilon}\right), \quad \varepsilon > 0.
        \end{align*}
    \end{assumption}

    We are interested in quantifying the strength of the correlations among particles when the scaling is given by \eqref{eq:scalingreg}.
    By denoting \(z=(x,v) \in \Lambda \times \R^d\) and \(\vc{z}_n =(z_1,\dots,z_n)\), we introduce the main statistical quantity that describes the state of the system. This is given by an absolutely continuous measure 
    \(W^{\varepsilon}\) on \(\Omega \coloneqq \cup_n \Omega_n\) where 
    \[
    \Omega_n\coloneqq\left\{ \vc{z}_n \in \left(\Lambda \times \R^d \right)^n\, |\, x_i \neq x_j \quad \text{for} \quad i \neq j\right\}.
    \]
    We assume \(W^{\varepsilon}\) admits a set of densities \(\{W^{\varepsilon}_n :\Omega_n \to \R_+\}_{n \geq 0}\) of Borel functions that are invariant under permutations and consider the following particular almost factorized ansatz.

    \begin{definition}[Maximally chaotic states]
    \label{def:mcs}
    Let \(f \colon \Lambda \times \R^d\to \R_+\) be a continuous probability density, and let \(\rho \colon \R^d \to \R_+\) be the corresponding spatial probability density, defined as
    \begin{align*}
        \rho(x) \coloneqq \int_{\R^d} f(x,v) \rmd v, \quad x \in \Lambda.
    \end{align*}
Assume further that \(\bar{\rho} \coloneqq \norm{\rho}_{L^\infty(\Lambda)} <+\infty\).
Given \(\varepsilon >0\) and \(\alpha\geq 0\), \(\delta >-d\), a \emph{maximally chaotic state} is a grand canonical measure with densities
    \begin{align}
    \label{eq:MCS}
        \frac{1}{n!}W_n^{\varepsilon}(\vc{z}_n)
        \coloneqq 
        \frac{1}{\mathcal{Z}_{\varepsilon}} \frac{\mu_{\varepsilon,\delta}^n}{n!}f^{\otimes n}(\vc{z}_n)\psi_n^{\varepsilon,\alpha}(\vc{x}_n), \quad \vc{z}_n \in \Omega_n,
    \end{align}
    where the partition function is defined as
    \[
    \mathcal{Z}_{\varepsilon}\coloneqq \sum_{n=0}^\infty\frac{\mu_{\varepsilon,\delta}^n}{n!}\int_{\Omega_{n}}f^{\otimes n}(\vc{z}_n)\psi_n^{\varepsilon,\alpha}(\vc{x}_n) \rmd \vc{z}_n
    \]
    and the functions \(\psi^{\varepsilon,\alpha}_n\) are the Boltzmann factors, defined as \(\psi_0^{\varepsilon,\alpha} \coloneqq 1\), \(\psi_1^{\varepsilon,\alpha} \coloneqq 1\), and for \(n \geq 2\) as
    \begin{align*}
    \psi^{\varepsilon,\alpha}_n(\vc{x}_n) \coloneqq \prod_{1\leq i < j \leq n} \rme^{-\beta \varepsilon^{\alpha}\varphi\left(\frac{x_{i}-x_j}{\varepsilon}\right)}, \quad \vc{x}_{n} \in \Omega_n, \quad \beta >0.
    \end{align*}
    Here, \(\beta\) is proportional to the \emph{inverse temperature} of the system. 
    \end{definition}
     \begin{remark}
    The activity parameter \(\mu_{\varepsilon,\delta}\) describes the asymptotic behaviour of the expected number of particles in the system. The precise asymptotic connection between \(\E_{\varepsilon}[\mathcal{N}]\) and \(\mu_{\varepsilon,\delta}\) is established in Lemma \ref{lemma:average-number} at the start of Section \ref{sec3}.
    \end{remark}

    The objects we want to study are given by correlation functions and truncated correlation functions (or cumulants). The latter objects quantify how far from the factorized state the \(j\)th correlation function is.

    \begin{definition}[Correlation functions, cumulants]
    \label{def:correlation-functions}
        Given a set of permutation invariant densities \(\{W_n^{\varepsilon}\}_{n=0}^\infty\) with respect to the corresponding Lebesgue measures, each respectively defined on \(\Omega_n\), we define the set of correlation functions \(\{\rho^{\varepsilon}_{j}\}_{j=0}^\infty\) through the formulas
        \begin{align*}
            \begin{cases}
                \rho_0^{\varepsilon}&\coloneqq 1\\
                \rho_j^{\varepsilon}(\vc{z}_j) &\coloneqq \sum_{n=0}^\infty \frac{1}{n!}\int_{\Omega_n}  W_{j+n}^{\varepsilon}(\vc{z}_{j+n})\rmd z_{j+1}\cdots \rmd z_{j+n},
            \end{cases}
        \end{align*}
        provided that the series on the left--hand-side converges.
        The corresponding \emph{truncated correlation functions} or \emph{cumulants} are defined recursively as follows. First, let
        \(\rho_{0}^{\varepsilon,c} \coloneqq 0\) and \(\rho_{1}^{\varepsilon,c} \coloneqq \rho_{1}^{\varepsilon}\), and then define the higher--order cumulants as
        \begin{align}
                \label{eq:cumulants}
        \rho_{j}^{\varepsilon, c}(\vc{z}_j) \coloneqq \rho_j^{\varepsilon}(\vc{z}_j) - \sum_{k=2}^j\sum_{\{\pi_1,\dots, \pi_k\} \vdash [j]} \prod_{i=1}^k \rho_{\abs{\pi_i}}^{\varepsilon,c}(\vc{z}_{\pi_i}), \quad j \geq 2.
        \end{align}
        In the above formula, the vector notation \(\vc{z}_A\) with \(A \subset [j] \coloneqq \{1,\dots, j\}\) refers to the point \((z_{a})_{a \in A} \in \Omega_{\abs{A}}\), and the notation \(\{\pi_1,\dots, \pi_k\} \vdash [j]\) means that the sets \(\pi_{i}\) partition \([j]\).

        For the convenience of the reader, we recall that the second and third order cumulants are given by
        \begin{align*}
        \rho^{\varepsilon,c}_2(z_1,z_2) &= \rho^{\varepsilon}_2(z_1,z_2) - \rho^{\varepsilon}_1(z_1)\rho^{\varepsilon}_1(z_2) \\
        \rho^{\varepsilon, c}_3(z_1,z_2,z_3)
        &= \rho^{\varepsilon}_3(z_1,z_2,z_3) - \rho^{\varepsilon}_2(z_1,z_2)\rho^{\varepsilon}_1(z_3) - \rho^{\varepsilon}_2(z_1,z_3)\rho^{\varepsilon}_1(z_2) \\
        &- \rho^{\varepsilon}_2(z_2,z_3)\rho^{\varepsilon}_1(z_1) + 2\rho^{\varepsilon}_1(z_1)\rho^{\varepsilon}_1(z_2)\rho^{\varepsilon}_1(z_3).
        \end{align*}
    \end{definition}

    We can now state the main theorem of this article. It establishes precise pointwise bounds
    for the truncated correlation functions (cumulants) introduced in \eqref{eq:cumulants},
    showing their decay in the scaling parameter \(\varepsilon\) for both finite range
    and Riesz--type potentials within the family of scalings presented in \eqref{eq:scalingreg}.

    For finite range potentials, the theorem yields exponential decay in terms of the minimal tree length \(\mathcal{L}(\vc{x}_j)\)
    connecting the macroscopic positions. The tree length \(\mathcal{L}(\vc{x}_j)\) is defined as the shortest tree connecting points \(\{x_{i}\}_{i=1}^j \subset \Lambda\), where one is allowed to introduce an arbitrary number of additional points \(x_{j+1},\dots, x_{j+n} \in \Lambda\). More precisely,
            \begin{align*}
                \mathcal{L}(\vc{x}_j) \coloneqq \inf_{n\geq 0}\inf_{x_{j+1}, \dots, x_{j+n} \in \Lambda}\min_{T \in \mathcal{T}_{j+n}} \sum_{\{\gamma, \gamma'\} \in E(T)} \abs{x_{\gamma}-x_{\gamma'}},
            \end{align*}
    where \(\mathcal{T}_{j+n}\) is the set of trees with \(j+n\) vertices and \(E(T)\) denotes the set of edges of \(T \in \mathcal{T}_{j+n}\).
    Note that by the triangle inequality,
            \begin{align*}
                \mathcal{L}(\vc{x}_j) \geq \max_{\gamma,\gamma' \in [j]}\abs{x_{\gamma}-x_{\gamma'}}.
            \end{align*}
    
    For Riesz potentials, the decay is expressed in terms of products along spanning trees of \emph{Mayer functions}, defined by
    \begin{equation}
        \label{eq:mayer}
        \eta^{\varepsilon,\alpha}(x,y)\coloneqq \rme^{-\beta \varepsilon^\alpha \varphi\left(\frac{x-y}{\varepsilon}\right)}-1,
    \end{equation}
    and leading in general only to polynomial decay. Notice that, if the domain \(\Lambda\) is bounded and the potential \(\varphi\) stable, one always has
    \begin{align*}
        \int_{\Lambda}  \abs{\eta^{\varepsilon,\alpha}(y,x)}\rmd x<+\infty, \quad  y \in \Lambda.
    \end{align*}
    
    \begin{theorem}[Pointwise estimates]
        \label{thm:main-theorem}
        Let $\varphi$ be a translation invariant pairwise potential satisfying the stability condition in \eqref{eq:stability}
        and let $\rho_j^{\varepsilon, c}$ be the $j$-th order cumulant of the maximally chaotic state introduced in \eqref{eq:MCS}.
        Given \(\alpha \geq 0\) and \(\delta >-d\), the following results hold:
        \begin{enumerate}
            \item \label{thm:pt1} Let $\varphi$ be finite range as in \eqref{eq:support} and assume that
            \begin{equation}
            \label{gamma_ass}
                \int \abs{\eta^{\varepsilon,\alpha}(x,0)} \rmd x \le C_{\beta, \varphi} \varepsilon^{\gamma},
            \end{equation}
            for some \(\gamma=\gamma(\alpha,d) \geq d\). Then, if
            \begin{equation}
            \label{eq:p-sr}
            p_{\mathrm{fr}}\coloneqq \gamma -d-\delta>0,
            \end{equation}
            there exist $\varepsilon_0(d,B,\beta,\gamma,\delta,\bar\rho,\varphi)>0$ and $C>0$ such that
            \[
            \dfrac{|\rho^{\varepsilon,c}_j(\vc z_j)|}{\mu_{\varepsilon,\delta}^j}
            \le
             \dfrac{(Cf)^{\otimes j}(\vc z_j)}{1-(\varepsilon/\varepsilon_0)^{p_{\mathrm{fr}}}}\, 
             j^{j-2} \,
             \left(\dfrac{\varepsilon}{\varepsilon_0}\right)^{p_{\mathrm{fr}}\left[\frac{\mathcal{L}(\vc x_j)}{\varepsilon}-(j-1)\right]_+}
            \]
            for every \(\varepsilon \in (0, \varepsilon_0)\), \(j \geq 2\) and \(\vc{z}_j \in (\Lambda \times \R^d)^j\).
            Moreover, in case \(p_{\mathrm{fr}}=0\), there exist \(\bar \rho_0(d,B, \beta, \varphi)>0\) and \(C>0\) such that
            \begin{equation}
            \label{eq:boundfr}
            \dfrac{|\rho^{\varepsilon,c}_j(\vc z_j)|}{\mu_{\varepsilon,\delta}^j}
            \le
             \dfrac{(Cf)^{\otimes j}(\vc z_j)}{1-\bar \rho/{\bar \rho}_0}\, 
             j^{j-2} \,
             \left(\dfrac{\bar \rho}{{\bar \rho}_0}\right)^{\left[\frac{\mathcal{L}(\vc x_j)}{\varepsilon}-(j-1)\right]_+}
            \end{equation}
            for every \(\bar\rho \in (0, {\bar \rho}_0)\), \(j \geq 2\) and \(\vc{z}_j \in (\Lambda \times \R^d)^j\).

            \medskip
            
            \item \label{thm:pt2} Assume \(\Lambda\) to be bounded and let \(\varphi\) be of Riesz--type as in \eqref{eq:Riesz} with power \(s >0\). Introducing the discrepancy \(\Delta_{s,d} \coloneqq s-d\),
            we set 
            \begin{align}
            \label{eq:plongrange}
                p_{\mathrm{ir}}\coloneqq \begin{cases}
                    \alpha + \Delta_{s,d}-\delta, \quad &\text{for} \quad 0<s\leq d,\\
                     \dfrac{d}{s}\alpha-\delta, \quad &\text{for} \quad s > d,
                \end{cases}
            \end{align}
            and assume that \(p_{\mathrm{ir}}>0\).
            Then, there exist \( \varepsilon_0(d, s, \alpha, \beta, \delta, \bar\rho)>0\) and \(C>0\) such that
            \begin{align}  
            \label{ub:ir-corr}
            \frac{\abs{\rho_{j}^{\varepsilon, c}(\vc{z}_j)}}{\mu_{\varepsilon, \delta}^j} \leq \frac{(C f)^{\otimes j}(\vc z_j)}{1-\kappa(\varepsilon,\varepsilon_0)}^{p_{\mathrm{ir}}}\, j^{j-2}\,\max_{T \in \mathcal{T}_{j}} \left(\prod_{\{\gamma, \gamma'\}\in E(T)} \abs{\eta^{\varepsilon,\alpha}(x_{\gamma}, x_{\gamma'})}\right),
            \end{align}
            for every \(\varepsilon \in (0, \varepsilon_0)\), \(j \geq 2\) and \(\vc{z}_j \in (\Lambda \times \R^d)^j\), where
            \[
            \kappa(\varepsilon,\varepsilon_0)\coloneqq \begin{cases}
                \dfrac{\varepsilon}{\varepsilon_0}\left(\dfrac{\log(\varepsilon)}{\log(\varepsilon_0)}\right)^{\frac{1}{p_{\mathrm{ir}}}}, \quad &\text{if} \quad s=d,\\[3mm]
                \dfrac{\varepsilon}{\varepsilon_0}\quad &\text{if} \quad s \neq d.
            \end{cases}
            \]
            Moreover, in case \(p_{\mathrm ir}=0\) and \(s \neq d\), there exist \(\bar \rho_0(d,s,\beta)>0\) and \(C>0\) such that
            \[
          \frac{\abs{\rho_{j}^{\varepsilon, c}(\vc{z}_j)}}{\mu_{\varepsilon, \delta}^j} \leq \frac{(C f)^{\otimes j}(\vc z_j)}{1- \bar \rho/\bar \rho_0}\, j^{j-2}\,\max_{T \in \mathcal{T}_{j}} \left(\prod_{\{\gamma, \gamma'\}\in E(T)} \abs{\eta^{\varepsilon,\alpha}(x_{\gamma}, x_{\gamma'})}\right)
            \]
            for every \(\bar\rho \in (0, {\bar \rho}_0)\), \(j \geq 2\) and \(\vc{z}_j \in (\Lambda \times \R^d)^j\).
        \end{enumerate}
        \end{theorem}

    \begin{remark}
    \label{rm1}
    We highlight a few points of relevance:
        \begin{enumerate}[]
            \item[(i)]\label{rm1:i} The exponent \(\gamma\) introduced in \eqref{gamma_ass} depends on the type of finite-range potential. For example, if \(\varphi \colon \R^d \to \bar{\R}\) is a hard-core potential such as
            \begin{align*}
            \varphi(x) =\begin{cases} +\infty, &\quad \text{ if }\abs{x} \leq 1 \\
                0, &\quad \text{ otherwise, }
            \end{cases}
            \end{align*}
            we get that 
            \[
            \int\abs{\eta^{\varepsilon,\alpha}(x,0)} \rmd x = \varepsilon^d \abs{B_1(0)},
            \]
            hence \(\gamma=d\).
            On the other hand, for potentials \(\varphi \in L^1\), we get the bound
            \begin{align*}
            \int \abs{\eta^{\varepsilon,\alpha}(x,0)} \rmd x &\le \varepsilon^d \left[\int_{\{\varphi\geq 0\}}\beta\, \varepsilon^\alpha\, \varphi(x) \rmd x+ \int_{\{\varphi<0\}}\beta\, \varepsilon^\alpha (-\varphi(x)) \, \rme^{2 \beta \varepsilon^\alpha B}\, \rmd x\right]\\
            &\le C_{\beta,\varphi} \, \varepsilon^{\alpha +d},
            \end{align*}
            where \(C_{\beta,\varphi}=\beta \, \rme^{2\beta B} \norm{\varphi}_{L^1} \) and so, in this case, \(\gamma=\alpha +d\).
            Finally, if \(\varphi(x)=|x|^{-s} \mathbf{1}_{|x|\le 1}\) with \(s>d\), by a direct computation we get that \(\gamma=d(1+\alpha/s).\)

         \item[(ii)] In the case of exponentially decaying long-range potentials of type \(\varphi(x)=\rme^{-|x|}\in L^1(\Lambda)\), the conclusion in \eqref{ub:ir-corr} still holds if \(\alpha-\delta >0\). This is valid since the interaction potential is in \(L^1\) and coincides with the finite-range case when \(\gamma=\alpha+d\) in \eqref{gamma_ass}.
    
         \item[(iii)]\label{rm:smallrho}
         As an alternative to the regime where \(\bar \rho\) is small, analogous pointwise estimates could be obtained in the threshold case \(p_{\rm fr}=0\) or \(p_{\mathrm{ir}}=0\) (with \(s \neq d\)) by taking the inverse temperature \(\beta\) sufficiently small, whenever the constants in the convergence criterion retain an explicit dependence on \(\beta\).
     
        \item[(iv)]
        For bounded finite-range potentials, the pointwise bound \eqref{eq:boundfr} can be improved by an additional factor \(\varepsilon^{\alpha(j-1)}\) for the \(j\)-th order cumulant.
        Indeed in this case, the stronger bound
        \[
        |\eta^{\varepsilon,\alpha}(x,y)|\le
        C_{\beta,\varphi}\, \varepsilon^\alpha \, \|\varphi\|_\infty,
        \]
        can be used to improve \eqref{ineq:improvable}. 
        The bounds obtained are consistent with the classical literature on correlation functions for Kac potentials \(\varphi_\gamma(q)=\gamma^d \varphi(\gamma q)\) in the limit \(\gamma \to 0\); see e.g. \cite{KacUhlenbeckHemmer1963, LebowitzStellBaerTheumann1966, Gates1975, GreweKlein1977, merola_asymptotic_1999}.\footnote{The correspondence between the Kac setting and our results is obtained by setting \(\varepsilon=\gamma^{\frac d\alpha}\), \(\alpha=\delta\), and microscopic scaling \(q=x \mu_{\varepsilon,\delta}^{\frac 1 d}\).}
        For example, for \(j=2\),  in the critical regime \(\alpha=\delta\) and for small values of the density \(\bar \rho\),
        our estimate gives an \(O(\gamma^d)\) bound for the two-point cumulants, consistently with the estimate established for the Lebowitz--Mazel--Presutti model in \cite[Theorem 11.5.7.2]{Presutti2009}.
      
        \item[(v)] In the infinite--range case we restrict our attention to repulsive Riesz potentials for simplicity and since their homogeneity makes the scaling dependence explicit. The tree--reduction method is not tied to this precise form, and everything can be adapted as is to algebraically decaying Mayer function given by repulsive, monotone potentials, provided that additionally the necessary integrability bounds hold uniformly in the scaling parameter. Moreover, with modifications to the step outlined in \ref{sec4.4}, the result is expected to hold for potentials with a long--range attractive part, provided that it decays fast enough.
    
        \item[(vi)] In Part \ref{thm:pt2} of Theorem \ref{thm:main-theorem}, the regime \(0<s\leq d\) is also allowed if the interaction is sufficiently weak and the gas sufficiently diluted so that \(p_{\rm ir}>0\), i.e.
            \[\alpha+s>d+\delta.\]
            
             \item[(vii)] In the finite-range case, the compact support of the interaction implies that all integrations are localized to balls of radius \(\varepsilon\), and the above estimates only require \(\rho \in L^\infty(\Lambda)\). By contrast, in the infinite-range case, the boundedness of the spatial domain \(\Lambda \subset \R^d\) is required due to the tails of the interaction, and for \(\Lambda= \R^d\) only the regime \(s>d\) gives integrable tails.
        \end{enumerate}
    \end{remark}
    
    \subsection{Comments on the scaling regimes}
    \label{subsection:comments-on-the-scaling-regimes}

    For a fixed dimension \(d\), the conditions on the parameters \(p_{\mathrm{fr}}\) and \(p_{\mathrm{ir}}\) in 
    \eqref{eq:p-sr} and \eqref{eq:plongrange} 
    determine the admissible regions where the cluster expansions converge and yield the decay of truncated correlations established in Theorem \ref{thm:main-theorem}.
    The boundaries of these regions, represented by solid blue lines in the figures below and called \emph{critical lines}, mark the onset of the non-ideal gas thermodynamics behaviour and corrections to chaoticity, as can be seen in Corollary \ref{corollary1}.

    In this section, we discuss the relationship between the admissible regions provided by the region of convergence of cluster expansions given in the previous theorem and the kinetic scaling regimes conjectured in \cite[Conjecture 1.2 and 1.3]{NotaVelazquezWinter2021} 
    which are expected to lead to kinetic equations of Boltzmann, Landau, or Balescu--Lenard type.
    These scalings are obtained 
    by imposing \(O(1)\) cumulative effect of particle interactions over a macroscopic time interval (cfr.~\cite{Spohn80, spohn_large_1991,Pulvirenti2021,NotaVelazquezWinter2021}).
    We distinguish between regimes dominated by the core-part of the potential and those determined by its tail.
    
    \subsubsection{Core behavior:} 
   Let
    \(\varphi(x)=|x|^{-s} \mathbf{1}_{|x|\le 1}\) be a Riesz-type cut-off potential. Here, the admissibility condition for the convergence of the cluster expansion of correlations is 
    \(
    p_{\mathrm fr}= \gamma-d-\delta>0,
    \)
    where \(\gamma\) depends on the power \(s\); see Remark~\ref{rm1}--(i). The corresponding critical line is given by 
    \[
    \delta_{\mathrm c}(\alpha,s,d):=\begin{cases}
        \alpha \quad &\text{for}\quad 0<s<d,\\[3mm]
        \dfrac ds \alpha \quad &\text{for} \quad s>d.
    \end{cases}
    \]
    For finite-range potentials of Riesz--type, the kinetic limit exhibit a transition between Landau and Boltzmann--type kinetic behaviors. In particular, in dimension \(d=3\), this transition occurs at the Coulomb exponent \(s=1\): if \(s<1\), sub-Coulombian potentials lead to a Landau equation, while for \(s>1\) super-Coulombian potentials lead to a Boltzmann equation. In case \(s=1\), a logarithmic Coulombian correction is needed; see \cite{LeBihanWinter2023, LeBihan2026}.
    The corresponding kinetic regimes in our parametrization, represented by a dashed line in Figures \ref{fig:kinline}--\ref{fig:kinline2} and called \emph{kinetic lines}, are given by
    \[
    \delta_{\mathrm k}(\alpha, s, d):=\begin{cases}
        2\alpha-1 \quad &\text{for}\quad 0<s<\frac{d-1}{2},\\[3mm]
        \dfrac{d-1}{s} \alpha-1 \quad &\text{for} \quad s>\frac{d-1}{2},
        \end{cases}
    \]
    where the \(s=\frac{d-1}{2}\) case requires a logarithmic correction. 
       Indeed, in the diffusive regime, since the potential \(\varphi\) varies on a scale \(\varepsilon\), the force will be \(O(\varepsilon^{\alpha-1})\) and act on a time interval \(O(\varepsilon)\). The variation of momentum due to a single scattering is \(O(\varepsilon^\alpha)\) and the number of particles met by a typical particle over a unit macroscopic interval of time is \(O(\varepsilon^{-\delta-1})\). In order to have a diffusive behaviour, a central limit theorem has to be valid and the total variance has to be \(O(1)\). This implies \(\delta=2\alpha-1\).
       In this case, for \(\alpha \in (0,1)\), the diffusive Landau equation is expected to be valid; see~\cite{BobylevPulvirentiSaffirio2013,VelazquezWinter2018Landau, Winter2021} for partial results on this largely open problem.
       
       In the collisional regime, we have \(O(1)\) variations of momentum due to a single scattering. This happens at a distance \(r_*\) such that \(\varepsilon^{\alpha+s} r_*^{-s}=O(1)\). Moreover, the total variation of momentum in a unit macroscopic time is given by \(\varepsilon^{-d-\delta} \, r_*^{d-1}\) and it has to be \(O(1)\). This implies \(\delta=\frac{d-1}{s}\alpha-1\). 
       The transition between these two regimes occurs at \(s=(d-1)/2\), corresponding to the threshold of integrability of the diffusive coefficient, given by the total contribution of momentum variances.
    For hard-core collisions and for non-rescaled potentials with \(\alpha=0\), we get \(\delta_{\mathrm k}=-1\) which is the low-density or Boltzmann--Grad scaling.
    The rigorous treatment of this scaling in the derivation of a dynamical description of the particle density of the system given by the Boltzmann equation has received significant attention since the 20th century, with notable milestones being, in case of hard-sphere particles, the short time--proof given in the \empty{Lanford's theorem}~\cite{ehlers_time_1975} and its subsequent improvements and extensions, leading up to the recent works~\cite{bodineau_long-time_2024,DengHaniMa2024} (see also~\cite{PulvirentiSaffirioSimonella2014, GallagherSaintRaymondTexier2014} for finite-range potentials).
    We note that, for \(\alpha\in (0,1)\), the kinetic line always lies below the critical line and, for potentials with \(s>(d-1)/2\), begins to decrease its slope as \(s\) increases and the potential becomes more singular near the origin; see Figure \ref{fig:kinline}.
    \begin{figure}[h]
    \centering
    \begin{subfigure}
        \centering 
        \includegraphics[width=0.90\linewidth]{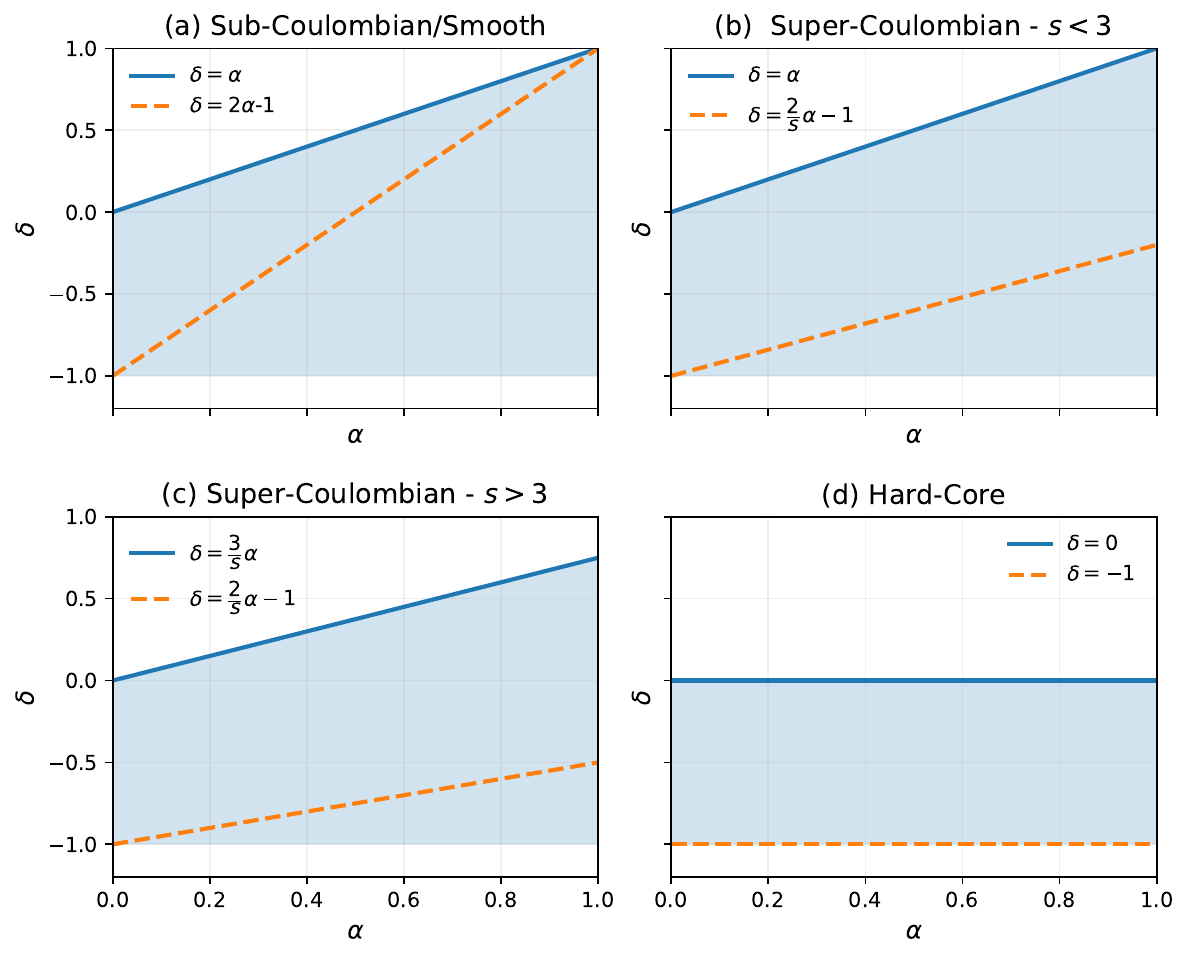}
        \caption{Admissible regions, critical and kinetic lines in the \((\alpha,\delta)\) plane for cut-off Riesz-type potentials \(\varphi(x)=|x|^{-s}\mathbf 1_{|x|\le 1}\) in \(d=3\). The solid blue line denotes the critical line given by the boundary \(p_{\mathrm fr}=0\). The shaded region below the critical line is the region of convergence of the cluster expansions. The dashed line is the kinetic line given by the scaling regimes for: (a) Regular or sub-Coulombian potential (b) Super-Coulombian with $s<3$, here \(s=5/2\); (c) Super-Coulombian potential with \(s>3\), here \(s=4\); (d) Hard-Core potential, formally corresponding to $s=\infty$.}
                \label{fig:kinline}
    \end{subfigure}
    \end{figure}
    \subsubsection{Tail behavior:}
    We consider the Riesz-type potentials for \(s<d\) and compare the kinetic predictions with the infinite-range result of Theorem \ref{thm:main-theorem}.
    In this case the critical line is given by
    \(
    \delta_{\mathrm c}(\alpha,s,d)=\alpha+s-d
    \)
    while the kinetic line is given for \(s>(d-1)/2\) by
    \[
    \delta_{\mathrm k}(\alpha,s,d)=\frac{d-1}{s}\alpha-1
    \]
    In particular,
    \[
    \delta_{\mathrm c}-\delta_{\mathrm k}=(s-d+1)(1+\frac \alpha s).
    \]
    Thus, there exists a region of admissibility if \(s\geq d-1\). For \(s=d-1\) we have that the two lines coincide, see Figure \ref{fig:kinline2}. This is the case of pure Manev interactions \(\varphi(x)=|x|^{-2}\), see e.g. \cite{BobylevDukesIllnerVictory1997}, where coexistence of mean-field and collisional effects are supposed to exist.\footnote{It is important to note that Manev interactions are typically considered in the case of an attractive potential, arising in gravitational contexts, while our analysis is focused on repulsive potentials.}
    For \(s<(d-1)/2\) many weak, long-range interactions dominate and collective effect become relevant. In this regime, the Landau equation is expected to require corrections due to the collective response of the medium, leading to the Balescu--Lenard equation; see e.g.~\cite{DuerinckxSaintRaymond2021,DuerinckxLeBihan2025}.
    \begin{figure}[h]
    \centering
    \begin{subfigure}
        \centering 
        \includegraphics[width=1.0\linewidth]{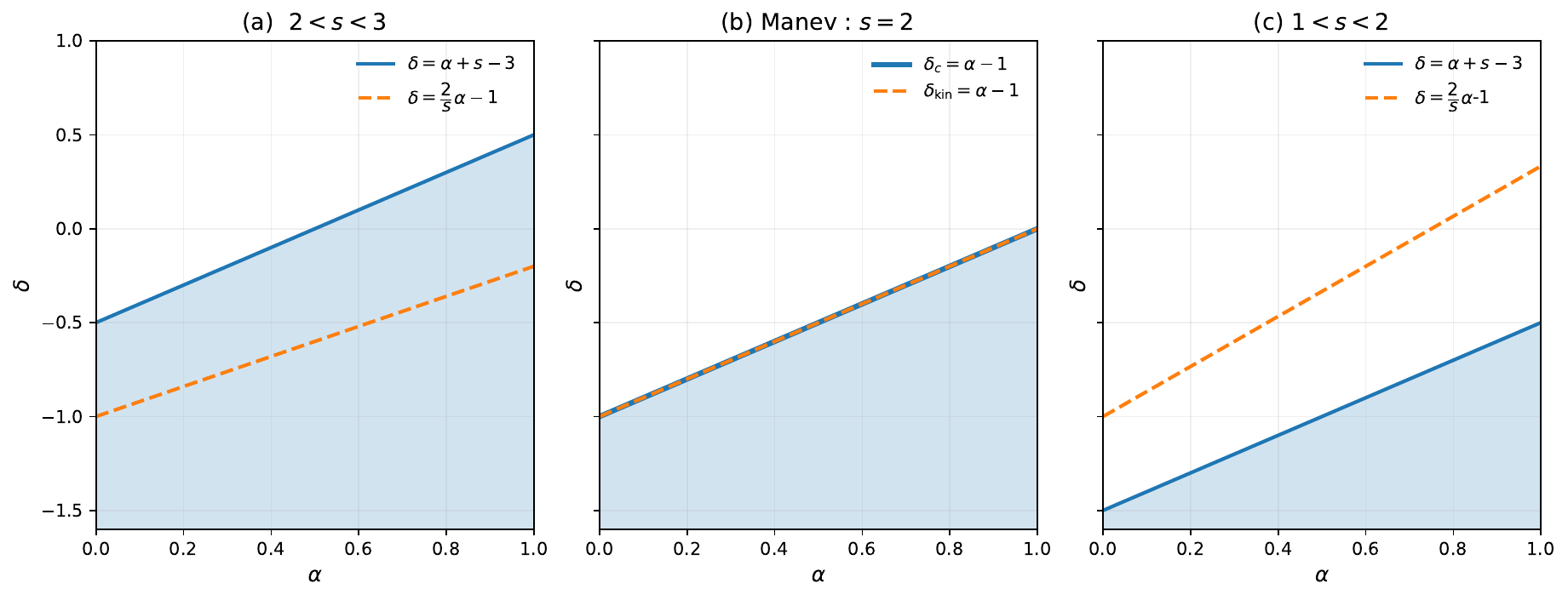}
        \caption{Admissible regions, critical and kinetic lines in the \((\alpha,\delta)\) plane for Riesz-type potentials \(\varphi(x)=|x|^{-s}\) in \(d=3\) with: (a) \(s=5/2\), (b) \(s=2\), pure Manev interaction, (c) \(s=3/2\). 
        }
                \label{fig:kinline2}
    \end{subfigure}
    \end{figure}
    
    \subsection{\(L^1\) theory and fluctuation results}
    \(L^1\) estimates for truncated correlation functions are more straightforward to derive, since the integrals appearing in the cluster expansion contain no fixed vertices, allowing one to integrate over all variables. The following proposition will be useful to prove Corollary \ref{corollary1}.
    
    \begin{proposition}[\(L^1\) estimates for cumulants]
    \label{prop:L1cum}
    Let \(\varphi\) be a translation invariant pairwise potential statisfying the stability condition in \eqref{eq:stability} and let \(\rho_{j}^{\varepsilon, c}\) be the \(j\)th order cumulant of the maximally chaotic state introduced in \eqref{eq:MCS}. Given \(\alpha \geq 0\) and \(\delta > -d\), the following results hold:
    \begin{enumerate}
        \item Let \(\varphi\) be finite range as in \eqref{eq:support} and assume that \eqref{gamma_ass} is valid
        for some \(\gamma = \gamma(\alpha, d) \geq d\). Then, if \(p_{\mathrm{fr}}>0\) as in  \eqref{eq:p-sr},
        there exist \(\varepsilon_0(d,B,\beta, \gamma, \delta, \bar{\rho},\varphi)>0\) and \(C > 0\) such that
        \begin{align*}
            \left \|\frac{\rho^\varepsilon_1}{\mu_{\varepsilon,\delta}}-f\right\|_{L^1(\R^d)} \leq \frac{\rme^{2\beta \varepsilon^\alpha B}}{1-(\varepsilon/\varepsilon_0)^{p_{\mathrm{fr}}}}(\varepsilon/\varepsilon_0)^{p_{\mathrm{fr}}}, \quad \forall \varepsilon \in (0, \varepsilon_0),
        \end{align*}
        and, for \(j\geq 2\),
        \begin{align*}
    \left\|\frac{\rho_{j}^{\varepsilon,c}}{\mu_{\varepsilon,\delta}^j}\right\|_{L^1((\R^d)^j)} \leq \frac{C^j\, \varepsilon^{\gamma(j-1)}}{1-(\varepsilon/\varepsilon_0)^{p_{\mathrm{fr}}}}j^{j-2},\quad \forall \varepsilon \in (0, \varepsilon_0).
        \end{align*}
        \item  Assume \(\Lambda\) to be bounded and let \(\varphi\) be of Riesz--type as in \eqref{eq:Riesz} with power \(s >0\).
            Then, if \(p_{\mathrm ir}>0\), there exist \( \varepsilon_0(d, s, \alpha, \beta, \delta, \bar\rho)>0\) and \(C>0\) such that
            \begin{align}
            \label{in:L1j1}
                \Big\|\frac{\rho_{1}^{\varepsilon, c}}{\mu_{\varepsilon, \delta}}-f\Big\|_{L^1(\R^d)} \leq \frac{(\varepsilon/\varepsilon_0)^{p_{\mathrm{ir}}}}{1-(\varepsilon/\varepsilon_0)^{p_{\mathrm{ir}}}}, \quad \forall \varepsilon \in (0, \varepsilon_0),
            \end{align}
            and, for \(j \geq 2\),
            \begin{align}   
            \label{in:L1j2}
            \left\|\frac{\rho_{j}^{\varepsilon, c}}{\mu_{\varepsilon, \delta}^j}\right\|_{L^1((\R^d)^j)} \leq \frac{C^{j}\, \omega(\varepsilon)^{j-1}}{1-(\varepsilon/\varepsilon_0)^{p_{\mathrm{ir}}}}\, j^{j-2}, \quad \forall \varepsilon \in (0, \varepsilon_0),
            \end{align}
            where 
              \begin{align}
        \label{eq:omega_epsilon}
        \omega(\varepsilon) \coloneqq \begin{cases}
            \varepsilon^{(s+\alpha)\frac ds},  &\text{if} \quad s>d,\\
            \varepsilon^{s+\alpha}\ln(\frac 1\varepsilon),  &\text{if} \quad s=d,\\
            \varepsilon^{s +\alpha} &\text{if} \quad s<d.
            \end{cases}
        \end{align}
    \end{enumerate}
        
    \end{proposition}

    We also establish the asymptotic behavior of the 
    \emph{fluctuation field} related to the maximally chaotic states \eqref{eq:MCS}. 
    For \(h \in C_b(\Omega_1)\), we define
    \begin{equation}
    \label{eq:fluctfield}
    \zeta_{\varepsilon}(h)\coloneqq\sqrt{\mu_{\varepsilon,\delta}}\left(\pi_{\varepsilon}(h)- \mathbb{E}_{\varepsilon}[\pi_\varepsilon (h)] \right), 
    \end{equation}
    where $\mathbb{E}_{\varepsilon}$ denotes the expectation with respect to the maximally chaotic distribution \eqref{eq:MCS} and the random empirical measure is given by
    \[
    \pi_{\varepsilon}(\rmd x,\rmd v)\coloneqq \dfrac{1}{\mu_{\varepsilon,\delta}}\sum_{i=1}^{\mathcal{N}} \delta_{x_i}(\rmd x)
    \otimes \delta_{v_i}(\rmd v),    \]
    with \(\mathcal{N}\) the random number of particles and \(\vc z_{\mathcal{N}}=(z_1,\dots,z_{\mathcal{N}}) \in\Omega=\cup_n \Omega_n\).

    \begin{corollary}
            \label{corollary1}
    Let \(\alpha \geq 0\), \(\delta>-d\) and consider the \emph{fluctuation field} \(\zeta_\varepsilon\) introduced in \eqref{eq:fluctfield}.
    Assume that the stable interaction potential \(\varphi\) is either finite-range as in \eqref{eq:support} or infinite-range as 
    in \eqref{eq:Riesz} and consider the finite/infinite-range parameters \(p_{\rm fr}, p_{\rm ir}\) introduced in \eqref{eq:p-sr} and \eqref{eq:plongrange}. The following results hold:
        \begin{enumerate}
        \item If \(p_{\rm fr}>0\) or \(p_{\rm ir}>0\), the \emph{fluctuation field} \(\zeta_\varepsilon\) converges in law to a centered Gaussian field \(\zeta\) with covariance
        \begin{align}
        \label{eq:limitcov}
            \mathbb{E}[\zeta(h)\zeta(g)] = \int_{\Omega_2}h(z_1)g(z_2)f(z_1)\delta(z_1-z_2)\rmd z_1 \rmd z_2, \quad h, g \in C_b(\Omega_1).
        \end{align}
    
    \item In the critical cases \(p_{\rm fr}=0\) or \(p_{\rm ir}=0\) with \(s \neq d\), in the regime 
    where \(\bar \rho\) is sufficiently small, the fluctuation field still converges in law to a centered Gaussian field \(\zeta\), whose covariance structure is given by
         \begin{align}
        \label{eq:limitcov2}
            \mathbb{E}[\zeta(h)\zeta(g)]& = \int_{\Omega_2}h(z_1)g(z_2)\left[ \bar \rho_1^c(z_1)\delta(z_1-z_2)+ \bar \rho^{c}_2(z_1,z_2)\right] \rmd z_1 \rmd z_2, 
        \end{align}
    with \(\bar \rho^{c}_1 \coloneqq \lim_{\varepsilon\to 0} \rho_1^{\varepsilon} \mu_{\varepsilon,\delta}^{-1}\) and \(\bar \rho^{c}_2 \coloneqq \lim_{\varepsilon\to 0} \rho_2^{\varepsilon,c}\mu_{\varepsilon,\delta}^{-1}\).
        \end{enumerate}
    \end{corollary}

    \section{Proofs of Theorem \ref{thm:main-theorem}, Proposition \ref{prop:L1cum} and Corollary \ref{corollary1}}
    \label{sec3}

    We start by proving that when \(p_{\mathrm{fr}}\) or \(p_{\mathrm{ir}}\) is positive, the average number of particles converges to \(\mu_{\varepsilon, \delta}\) as \(\varepsilon \to 0\).
    \begin{lemma}[Average number of particles]
    \label{lemma:average-number}
        Let \(\varphi\) be a stable interaction potential as described in \eqref{eq:support} or \eqref{eq:Riesz}. Then, if \(p_{\rm fr}>0\) (respectively, if \(p_{\rm ir}>0\)), it holds that
        \begin{equation}
        \label{sec3:avnum}
                    \lim_{\varepsilon\to0} \dfrac{\E_\varepsilon[\mathcal{N}]}{\mu_{\varepsilon,\delta}}=1.
        \end{equation}
    \end{lemma}

    \begin{proof}
    We give a proof for interaction potentials of Riesz--type. 
        By definition, we have the expression
        \[
        \dfrac{\E_\varepsilon[\mathcal{N}]}{\mu_{\varepsilon,\delta}}=
        \dfrac{\sum_{n \geq 0} \frac{1}{n!}\mu_{\varepsilon,\delta}^n A(n+1,\varepsilon)}{\sum_{n \geq 0} \frac{1}{n!}\mu_{\varepsilon,\delta}^n\,A(n,\varepsilon)},
        \]
        where
        \[
        A(n,\varepsilon):=\int_{\Omega_{n}}\prod_{1\leq i<j\leq n} \exp\left\{- \beta \varepsilon^\alpha\varphi\left(\frac{x_i-x_j} {\varepsilon}\right)\right\} \rmd\rho^{\otimes n}( {\vc x}_{n}).
        \]
        Since \(\varphi\geq0\), we have immediately that \(A(n+1,\varepsilon)\le A(n,\varepsilon)\), this implies that \(\E_\varepsilon[\mathcal{N}]\leq \mu_{\varepsilon,\delta}\).
        For a lower bound, we write the Boltzmann factors in terms of Mayer functions:
        \[
         \exp\left\{- \beta\varepsilon^\alpha\varphi\left(\frac{x_i-x_j} {\varepsilon}\right)\right\}=1-\abs{\eta^{\varepsilon,\alpha}(x_i,x_j)},
        \]
        and then use the elementary inequality \(\prod_{1\leq i\leq n}(1-\abs{\eta^{\varepsilon,\alpha}(x_i,x_{n+1})})\geq 1-\sum_{i=1}^n \abs{\eta^{\varepsilon,\alpha}(x_i,x_{n+1})}\). This implies 
        \[
        A(n+1,\varepsilon)\geq A(n,\varepsilon) - \bar \rho\,  nA\,(n,\varepsilon) \int\abs{\eta^{\varepsilon,\alpha}(0,x)} \rmd x.
        \]
        Hence, the ratio of the expected number of particles and \(\mu_{\varepsilon,\delta}\) is bounded from below as
        \[
          \dfrac{\E_\varepsilon[\mathcal{N}]}{\mu_{\varepsilon,\delta}} \geq 1- \bar \rho \mu_{\varepsilon,\delta} \int\abs{\eta^{\varepsilon,\alpha}(0,x)} \rmd x\geq 1- \bar \rho \, C_{\beta,s}\mu_{\varepsilon,\delta} \omega(\varepsilon)=1-\bar \rho \, C_{\beta,s} \varepsilon^{p_{\rm ir}}.
        \]
        This implies \eqref{sec3:avnum}. 
    \end{proof}

    \begin{proof}[Proof of Theorem \ref{thm:main-theorem}]
      We divide the proof in the following steps. In Step 1 we recall how to obtain explicit representation formulas for the partition function and the truncated correlation functions. In Step 2 we study the decay of truncated correlation functions in the case of short-range potentials, while in Step 3 we consider the case of long-range ones. 

    \emph{Step 1: Representation formulas.} 
   Consider the partition function
    \begin{equation}
    \label{eq:partfun}
        \mathcal{Z}_\varepsilon= \sum_{n \geq 0} \dfrac{\mu_{\varepsilon,\delta}^n}{n!}\int_{\Omega_n} f^{\otimes n}(\vc{z}_n) \psi^{\varepsilon, \alpha}_n(\vc x_n) \,{\rm d}\vc z_n,
    \end{equation}
    and the truncated correlation functions \(\rho^{\varepsilon,c}_j\) introduced in \eqref{eq:cumulants}.
    By a standard combinatorial argument (see e.g. \cite{SimonellaWinter2026}), we can show that
    \begin{equation}
    \label{formula:partfun}
        \mathcal{Z}_\varepsilon=\exp \left\{\mu_{\varepsilon, \delta} \sum_{n=1}^\infty \dfrac{\mu_{\varepsilon, \delta}^{n-1}}{n!} \int_{\Lambda^n} 
        u^{\varepsilon, \alpha}_n(\vc{x}_n){\rmd} \rho^{\otimes n}(\vc x_n)
        \right\},
    \end{equation}
    and
    \begin{equation}
    \label{formula:truncorr}
    \frac{\rho^{\varepsilon,c}_j({\vc z}_j)}{\mu_{\varepsilon,\delta}^j}
        =
    f^{\otimes j}(\vc z_j)
    \sum_{n=0}^{\infty}
    \frac{\mu_{\varepsilon,\delta}^n}{n!}
    \int_{\Lambda^n}
    u^{\varepsilon,\alpha}_{j+n}(\vc x_{j+n})
    \, \rmd \rho(x_{j+1}) \cdots \rmd \rho(x_{j+n}).
    \end{equation}
   Here, recalling the definition of Mayer function $\eta^{\varepsilon,\alpha}(x,y)\coloneqq e^{-\beta \varepsilon^\alpha \varphi(\frac{x-y}{\varepsilon})}-1$,
    we have introduced the Ursell functions 
     \[
        u^{\varepsilon, \alpha}_{n}(\vc{x}_{n})\coloneqq \sum_{G \in \mathcal{C}_{n}} \prod_{\{\gamma, \gamma'\}\in E(G)}\eta^{\varepsilon,\alpha}(x_\gamma, x_{\gamma'}), \quad \vc{x}_n \in \Lambda^n,
        \]
    where $\mathcal{C}_n$ denotes the set of connected graphs $G$ with $n$ vertices \([n] \coloneqq \{1,\dots,n\}\). The sets of vertices and edges in \(G\) are denoted by $V(G)$ and $E(G)$, respectively.
    
    Formula \eqref{formula:partfun} can be deduced by expanding the Boltzmann factors $\psi^{\varepsilon, \alpha}_n$ in \eqref{eq:partfun}, obtaining that
        \[
        \psi^{\varepsilon, \alpha}_n(\vc{x}_n)= \prod_{1 \le \gamma<\gamma'\le n} \left(1+ \eta^{\varepsilon, \alpha}(x_\gamma, x_{\gamma'})\right)=
        \sum_{G \in \mathcal{G}_n} \prod_{\{\gamma, \gamma'\}\in E(G)} \eta^{\varepsilon,\alpha}(x_\gamma, x_{\gamma'}),
        \]  
        where $\mathcal{G}_n$ denotes the set of graphs $G$ with $n$ vertices \(\{1,\dots, n\}\).
        Reorganizing the sum by isolating the connected graph-components, we obtain that
        \begin{equation}
        \label{eq:connectedgraph}
        \psi^{\varepsilon, \alpha}_n(\vc{x}_n)= \sum_{k=1}^n \sum_{\{\pi_1,\dots,\pi_k\} \vdash {[n]}} \prod_{i=1}^k u^{\varepsilon,\alpha}_{|\pi_i|}(\vc{x}_{\pi_i}),
        \end{equation}
        where we recall that, as in \eqref{eq:cumulants}, the notation $\{\pi_1,\dots,\pi_k\} \vdash {[n]}$ stands for the possible nontrivial partitions of the set $[n]$.
        The expression of the partition function follows by substituting \eqref{eq:connectedgraph} inside \eqref{eq:partfun} and reorganizing the sum.
        A similar procedure gives the representation formula in \eqref{formula:truncorr}.

        The convergence of the series appearing in \eqref{formula:partfun} and \eqref{formula:truncorr} can be established for small values of \(\varepsilon\). Since the argument is analogous to those presented in the next two steps, we omit the details here. 
        A key ingredient is
        the following lemma, which estimates 
        cluster expansions over connected graphs of order \(k\) in terms of cluster expansions over trees of the same order.

        \begin{lemma}[Tree-graph inequality (\cite{Penrose1967}, \cite{PoghosyanUeltschi2009}, \cite{Ueltschi2017})]
        \label{lem:tree_graph}
        
        Assume that the interaction potential satisfies the stability condition. Then, for all $k \geq 1$ and all configurations $\vc x_k \in \Lambda^k$, one has
        \begin{equation}
            \label{eq:tree_new}
        \left|
            \sum_{G \in \mathcal{C}_k}
\prod_{\{\gamma,\gamma'\} \in E(G)}
\eta^{\varepsilon,\alpha}(x_\gamma, x_{\gamma'})
\right|
\leq
\rme^{2k\beta \varepsilon^\alpha B}
\sum_{T \in \mathcal{T}_k}
\prod_{\{\gamma,\gamma'\} \in E(T)}
\big|\eta^{\varepsilon,\alpha}(x_\gamma, x_{\gamma'})\big|,
\end{equation}
where \(\mathcal{T}_k\) denotes the set of labeled trees with \(k\) vertices.
\end{lemma}

        \emph{Step 2: Finite-range case.} 
    Starting from \eqref{formula:truncorr} and using the tree-graph inequality \eqref{eq:tree_new}, we obtain
    \begin{equation}
    \label{eq:bound_start}
    \left|
    \frac{\rho^{\varepsilon,c}_j(\vc z_j)}{\mu_{\varepsilon,\delta}^j}
\right|
\le
f^{\otimes j}(\vc  z_j)\, \rme^{2j\beta  \varepsilon^\alpha B}
\sum_{n\ge 0}
\frac{\mu_{\varepsilon,\delta}^n}{n!}
\rme^{2n\beta \varepsilon^\alpha B}
\sum_{T\in\mathcal T_{j+n}}
\inI{T}^{\varepsilon,\alpha}(\vc x_j)
\end{equation}
where,
for fixed $n$ and $T\in\mathcal T_{j,n}$, 
we introduced
\begin{equation}
\label{int_graph}
\inI{T}^{\varepsilon,\alpha}(\vc x_j)
\coloneqq
\int_{\Lambda^n}
\prod_{\{\gamma,\gamma'\}\in E(T)}
\big|\eta^{\varepsilon,\alpha}(x_\gamma,x_{\gamma'})\big|
\, \rmd\rho(x_{j+1}) \cdots \rmd\rho(x_{j+n}).
\end{equation}
Since the interaction potential \(\varphi(\abs{\cdot})\) is supported in \(B(0,1)\), the Mayer function $\eta^{\varepsilon,\alpha}(x,y)$ vanishes unless $|x-y|\le \varepsilon$. Hence the integration can be restricted to configurations such that all edges of the tree connect points at distance of order $\varepsilon$. 
In particular, $\inI{T}^{\varepsilon,\alpha}(\vc x_j)=0$ unless $n\ge n_0(\vc x_j)$, where $n_0$ denotes the minimal number of balls of radius $\varepsilon$ required to connect $\{x_1,\dots,x_j\}$.

We now estimate $\inI{T}^{\varepsilon,\alpha}$ by exploiting the tree structure. Using Fubini’s theorem, we integrate iteratively along the tree:
\begin{enumerate}
    \item[(i)] first over the leaves (vertices of degree one), thereby reducing the size of the tree;
\item[(ii)] then repeating the procedure until only $j-1$ chains connecting the fixed points remain;
\item[(iii)] finally integrating along these chains.
\end{enumerate}
By \eqref{gamma_ass},
each elementary integration is controlled by
\begin{equation*}
\int |\eta^{\varepsilon,\alpha}(0,x)|\, \rmd\rho(x)
\le \, C_{\beta, \varphi}\,\bar \rho \, \varepsilon^{\gamma}.
\end{equation*}
Similarly, a chain with $m+1$ edges is bounded by
\(
(1+\rme^{2\beta \varepsilon^\alpha B}) \big(C_{\beta,\varphi}\,\bar \rho\, \varepsilon^{\gamma}\big)^m.
\)
As a consequence, we obtain
\begin{align}
\label{ineq:improvable}
\inI{T}^{\varepsilon,\alpha}(\vc x_j)
\le
(1+\rme^{2\beta  \varepsilon^\alpha B})^{j-1}
\big(C_{\beta, \varphi}\, \bar\rho\, \varepsilon^{\gamma}\big)^n.
\end{align}

Plugging the estimate \eqref{ineq:improvable} into \eqref{eq:bound_start}, we deduce that the rescaled truncation correlation functions are bounded as
\begin{align*}
\left|
\frac{\rho^{\varepsilon,c}_j(\vc z_j)}{\mu_{\varepsilon,\delta}^j}
\right|
&\le
f^{\otimes j}(\vc z_j)\,
\rme^{2j\beta \varepsilon^\alpha B}
(1+\rme^{2\beta \varepsilon^\alpha B})^{j-1}
\sum_{n\ge n_0}
\frac{1}{n!}
\big(C_{\beta,\varphi}\,\bar \rho\, \mu_{\varepsilon,\delta}\, \varepsilon^{\gamma}\big)^n
\rme^{2n\beta \varepsilon^\alpha B}
|\mathcal T_{j+n}|.
\end{align*}
Using Cayley’s formula $|\mathcal T_k|=k^{k-2}$ for the number of labeled trees together with the combinatorial bound
\begin{equation*}
\frac{(j+n)^{j+n-2}}{n!}
\le
j^{j-2}C^j (2\rme)^n,
\end{equation*}
allows us to infer
\begin{align*}
\left|
\frac{\rho^{\varepsilon,c}_j(\vc z_j)}{\mu_{\varepsilon,\delta}^j}
\right|
\le
f^{\otimes j}(\vc z_j)\, C^j(1+\rme^{2\beta \varepsilon^\alpha B})^j
j^{j-2}
\sum_{n\ge n_0}
\big(2\,C_{\beta,\varphi}\,\bar \rho \, \rme^{2\beta \varepsilon^\alpha B+1}\, \mu_{\varepsilon,\delta}\, \varepsilon^{\gamma}\big)^n.
\end{align*}
Finally, observing that
\(
\mu_{\varepsilon,\delta}\, \varepsilon^{\gamma}
=
\varepsilon^{\gamma -d -\delta},
\)
we conclude that the series is absolutely convergent for \(p_{\mathrm fr}=\gamma -d-\delta>0\) and $\varepsilon<\varepsilon_0$ where
\[
\varepsilon_0\coloneqq \left(\dfrac{1}{{C'_{\beta,\varphi}}\,\bar \rho}\right)^{\frac{1}{\gamma -d-\delta}}\wedge 1,\]
with \(C_{\beta,\varphi}':= 2\, \rme^{2\beta B+1}\,C_{\beta,\varphi}\). If \(p_{\mathrm fr}=0\), the series is absolutely convergent for \(\bar \rho < \bar\rho_0\) where
\(
\bar\rho_0:=1/C'_{\beta,\varphi}.
\)

\emph{Step 3: Infinite-range case.}        
       Our aim is to bound the analogous expression of \eqref{eq:bound_start}. Since \(B=0\) for Riesz--type potentials, we get
    \begin{align*}
        \label{ineq:long-range-cluster-expansion}
        \frac{\abs{\rho_j^{\varepsilon, c}(\vc{z}_j)}}{\mu_{\varepsilon, \delta}^j} \leq f^{\otimes j}(\vc{z}_j) 
        \sum_{n=0}^{\infty} \frac{\mu_{\varepsilon, \delta}^{n}}{n!} 
        \sum_{T \in \mathcal{T}_{j,n}} \inI{T}^{\varepsilon,\alpha}(\vc{x}_j), \quad \vc{z}_j = (\vc{x}_j, \vc{v}_j) \in \Lambda^j \times \R^j,
    \end{align*}
    where \(\inI{T}^{\varepsilon,\alpha}(\vc{x}_j)\) is defined in \eqref{int_graph}.
    To this end, we state a proposition that estimates, through a precise graph reduction procedure, the integral quantity \(\inI{T}^{\varepsilon,\alpha}(\vc{x}_j)\). This result will be proved in the following Section \ref{sec4}.
    
    \begin{proposition}[Graph Reduction Procedure]
    \label{graph_red}
        Given \(\vc x_j \in \Lambda^j\) and \(T \in \mathcal{T}_{j,n}\), it holds that:
        \[
        \inI{T}^{\varepsilon,\alpha}(\vc{x}_j) \le \rme ^{2(j-1)/\rme}\,\left(C_{\beta, s} \bar \rho\, \omega(\varepsilon)\right)^{n}\max_{T^* \in \mathcal{T}_{j}} \left(\prod_{\{\gamma, \gamma'\} \in E(T^*)}\abs{\eta^{\varepsilon,\alpha}(x_{\gamma}, x_{\gamma'})}\right),
        \]
        where \(C_{\beta,s}\) is a given constant and \(\omega(\varepsilon)\) is defined in \eqref{eq:omega_epsilon}.
    \end{proposition}

    By Proposition \ref{graph_red} and using the inequality \(\frac{(j+n)^{j+n-2}}{n!} \leq j^{j-2}C^{j}(2\rme)^n\), we obtain the following upper bound
    \begin{multline}
        \label{ineq:upper-bound-truc}
        \frac{\abs{\rho_{j}^{\varepsilon, c}(\vc{z}_j)}}{\mu_{\varepsilon, \delta}^j}
        \leq 
        f^{\otimes j}(\vc{z}_j) 
        j^{j-2}C^{j} 
        \sum_{n=0}^{\infty}
      (2 {\rme} \, C_{\beta, s} \, \bar \rho\, \omega(\varepsilon)\mu_{\varepsilon,\delta} 
      )^n 
      \max_{T \in \mathcal{T}_{j}} \left(\prod_{\{\gamma, \gamma'\} \in E(T)}\abs{\eta^{\varepsilon,\alpha}(x_{\gamma}, x_{\gamma'})}\right).
     \end{multline}
     Recalling the definition of \(\omega(\varepsilon)\) in \eqref{eq:omega_epsilon}, if \(p_{\mathrm ir}>0\) and \(s \neq d\),
     the series on the r.h.s. of \eqref{ineq:upper-bound-truc} is convergent, when \(\varepsilon>0\) satisfies
     \begin{align}
        \label{ineq:plr-radius-of-convergence}
         \varepsilon < \varepsilon_0 \coloneqq \left(\frac{1}{C'_{\beta, s} \, \bar \rho}\right)^{\frac{1}{p_{\mathrm{ir}}}},
     \end{align}
     where \(C'_{\beta,s}\coloneqq 2 {\rme} \, C_{\beta, s} \,  
     \) and \(p_{\mathrm{ir}}\) is defined in \eqref{eq:plongrange}.
     If \(s=d\), a logarithmic correction appears and the series on the r.h.s. of \eqref{ineq:upper-bound-truc} is
     convergent for \(\varepsilon<\varepsilon_0\), where \(\varepsilon_0\) is sufficiently small so that 
     \begin{equation}
     \label{ineq:plr-rof}
    C'_{\beta,s} \bar \rho\, \varepsilon_0^{p_{\mathrm ir}}\,\ln(1/\varepsilon_0)<1.
     \end{equation}
     When \eqref{ineq:plr-radius-of-convergence}
     and \eqref{ineq:plr-rof} are satisfied, we get the claimed upper bounds. The case \(p_{\mathrm ir}=0\) with \(s \neq d\) is similar to the analogous finite-range case and can be treated by using the smallness of \(\bar \rho\).
     \end{proof}

    We now proceed to prove Proposition \ref{prop:L1cum}, which establishes \(L^1\)--bounds of the rescaled cumulants.
  \begin{proof}[Proof of Proposition \ref{prop:L1cum}]
        We give the proof for finite-range stable potentials as in \eqref{eq:support}, since the proof for the Riesz--type case is obtained by imposing \(B=0\) and substituting \(p_{\rm fr}\) with \(p_{\rm ir}\).
        
        We start from the representation formula in \eqref{formula:truncorr}. If \(j=1\), by the tree-graph inequality in \eqref{lem:tree_graph}, we get
           \begin{align*}
            \left \|\frac{\rho^\varepsilon_1}{\mu_{\varepsilon,\delta}}-f\right\|_{L^1}
            &\le
            \, \rme^{2\beta  \varepsilon^\alpha B}
\sum_{n\ge 1}
\frac{\mu_{\varepsilon,\delta}^n}{n!}
e^{2n\beta \varepsilon^\alpha B}
\sum_{T\in\mathcal T_{n+1}}
\int_{\Lambda^{n+1}}
\prod_{\{\gamma,\gamma'\}\in E(T)}
\big|\eta^{\varepsilon,\alpha}(x_\gamma,x_{\gamma'})\big|
\, {\rm d}\rho^{\otimes n+1}(\vc x_{n+1}).
            \end{align*}
   Since there are no fixed vertices, we bound the tree integral
   using Fubini's theorem iteratively, starting from the leaves of the tree. There are \(n\) edges in a tree of order \(n+1\) and integration of each edge produces as an upper bound the factor \(\bar{\rho}C_{\beta, \varphi}\varepsilon^{\gamma}\). With the help of \eqref{gamma_ass}, the integral can be bounded as
    \begin{align*}
        \int_{\Lambda^{n+1}}\prod_{\{\gamma,\gamma'\}\in E(T)}
\big|\eta^{\varepsilon,\alpha}(x_\gamma,x_{\gamma'})\big| \rmd \rho^{\otimes (n+1)}(\vc{x}_{n+1}) \leq \prod_{j=1}^{n} \int_{\Lambda} \abs{\eta^{\varepsilon,\alpha}(0,x)} \rmd \rho(x) \leq (\bar{\rho}\, C_{\beta,\varphi}\,\varepsilon^\gamma)^{n}.
    \end{align*}
    
    Using \(\frac{(n+1)^{n-1}}{n!} \leq \rme^n\), we obtain the following upper bound
    \begin{equation*}
         \left \|\frac{\rho^\varepsilon_1}{\mu_{\varepsilon,\delta}}-f\right\|_{L^1}
        \leq 
     \rme^{2\beta B}
\sum_{n=1}^{\infty}
      ( \rme^{2\beta B+1} \, C_{\beta, \varphi} \, \bar \rho\, \varepsilon^\gamma\mu_{\varepsilon,\delta} 
      )^n.
     \end{equation*}
     Therefore, since \(\varepsilon^{\gamma}\mu_{\varepsilon,\delta} = \varepsilon^{p_{\mathrm{fr}}}\), the series on the r.h.s. is convergent when \(\varepsilon>0\) satisfies
     \begin{align*}
         \varepsilon < \varepsilon_0 \coloneqq \left(\frac{1}{\rme^{2\beta B +1} \, C_{\beta,\varphi}  \, \bar \rho}\right)^{\frac{1}{p_{\mathrm{fr}}}},
     \end{align*}
     where the exponent \(p_{\mathrm{fr}}\) is defined in \eqref{eq:p-sr}.

    If \(j \geq 2\), we have by means of \eqref{formula:truncorr}
        \begin{align*}
            \left \|\frac{\rho^\varepsilon_j}{\mu_{\varepsilon,\delta}^j}\right\|_{L^1}
            &\le
            \, \rme^{2\, j\,\beta  \varepsilon^\alpha B}
\sum_{n\ge 0}
\frac{\mu_{\varepsilon,\delta}^n}{n!}
e^{2n\beta \varepsilon^\alpha B}
\sum_{T\in\mathcal T_{n+j}}
\int_{\Lambda^{n+j}}
\prod_{\{\gamma,\gamma'\}\in E(T)}
\big|\eta^{\varepsilon,\alpha}(x_\gamma,x_{\gamma'})\big|
\, {\rm d}\rho^{\otimes n+j}(\vc x_{n+1})\\
& \le \rme^{2\, j\,\beta  \varepsilon^\alpha B}
\sum_{n\ge 0}
\frac{\mu_{\varepsilon,\delta}^n}{n!}
e^{2n\beta \varepsilon^\alpha B}
\abs{\mathcal T_{n+j}}(C_{\beta,\varphi}\bar \rho \, \varepsilon^\gamma)^{n+j-1}.
            \end{align*}
    Using Cayley’s formula $|\mathcal T_k|=k^{k-2}$ together with the combinatorial bound
\begin{equation*}
\frac{(j+n)^{j+n-2}}{n!}
\le
j^{j-2}C^j (2\rme)^n,
\end{equation*}
allows us to infer
 \begin{align*}
            \left \|\frac{\rho^\varepsilon_j}{\mu_{\varepsilon,\delta}^j}\right\|_{L^1}
            \le \rme^{2\, j\,\beta  \varepsilon^\alpha B}\,
            j^{j-2}\, C^j\, (C_{\beta,\varphi}\, \bar \rho \, \varepsilon^\gamma)^{j-1}
\sum_{n\ge 0}
(2\mu_{\varepsilon,\delta}\, \varepsilon^\gamma\, \rme^{2\beta B+1}C_{\beta,\varphi} \bar \rho)^n,
            \end{align*}
from which the thesis follows.
    \end{proof}

    We now pass to the proof of the Corollary \ref{corollary1} concerning the fluctuations of the limit field, which we stated at the end of Section \ref{sec2}.
     \begin{proof}[Proof of Corollary \ref{corollary1}]
    We present the proof only for the long-range interaction case, the short-range case can be treated by similar arguments. 
    The covariance of the fluctuation field \(\zeta_{\varepsilon}\), defined in \eqref{eq:fluctfield}, is given by (see, e.g., \cite[p.881]{Bodineau2020}):
        \begin{align*}
            \E_\varepsilon[\zeta_{\varepsilon}(h)\zeta_{\varepsilon}(g)] &=\mu_{\varepsilon, \delta} \mathrm{Cov}[\pi_\varepsilon(h),\pi_\varepsilon(g)]=\frac{1}{\mu_{\varepsilon,\delta}}\mathrm{Cov}\left[\sum_{i=1}^\mathcal{N}h(x_i),\sum_{j=1}^\mathcal{N}g(x_j)\right] \\
            &=\frac{1}{\mu_{\varepsilon,\delta}}\left(\int_{\Omega_1}h(z_1)g(z_1)\rho_1^{\varepsilon}(z_1)\rmd z_1
            +\int_{\Omega_2} h(z_1)g(z_2)\rho_{2}^{\varepsilon, c}(z_1,z_2) \rmd z_1 \rmd z_2\right),
            \end{align*}
            where \(\rho_1^\varepsilon, \rho_2^{\varepsilon,c}\) are the truncated correlation functions introduced in \eqref{eq:cumulants}.

            If \(p_{\rm ir}>0\), by the \(L^1\)-estimate \eqref{in:L1j1}, it holds that \(\rho_1^\varepsilon/\mu_{\varepsilon,\delta} \to f\) in \(L^1\). This implies that 
            \[
            \lim_{\varepsilon \to 0}\frac{1}{\mu_{\varepsilon,\delta}}\int_{\Omega_1}h(z_1)g(z_1)\rho_1^{\varepsilon}(z_1)\rmd z_1 =\int_{\Omega_1} h(z_1) g(z_1) f(z_1) \rmd z_1.
            \]
            For the second-order cumulant, by \eqref{in:L1j2}, we have
            \begin{align*}
                   \left \|\frac{\rho^{\varepsilon,c}_2}{\mu^2_{\varepsilon,\delta}}\right\|_{L^1} \le \frac{C\, \omega(\varepsilon)}{1-(\varepsilon/\varepsilon_0)^{p_{\mathrm{ir}}}}
            \end{align*}
            where \(\omega(\varepsilon)\) is defined in \eqref{eq:omega_epsilon} and \(\varepsilon<\varepsilon_0\). Hence, we get
            \begin{align*}
        \mu_{\varepsilon,\delta}\int_{\Omega_2}h(z_1)g(z_2)\frac{\rho^{\varepsilon,c}_2(z_1,z_2)}{\mu_{\varepsilon, \delta}^2}\rmd z_1 \rmd z_2
            \le \frac{C}{1-(\varepsilon/\varepsilon_0)^{p_{\mathrm{ir}}}} \|h\|_\infty \|g\|_\infty \varepsilon^{p_{\rm ir}},
        \end{align*}
        and the contribution is infinitesimal for \(p_{\rm ir}>0\). 
        We now discuss the case \(p_{\rm ir}=0\) with \(s \neq d\): we expand in clusters the rescaled first and second order cumulants.  
        As discussed in Theorem~\ref{thm:main-theorem}, this representation is valid for sufficiently small values of \(\bar \rho\), and it converges absolutely for such values. We characterize the limit behavior of the individual terms: from \eqref{formula:truncorr}, each term is given by
        \[
        \mu_{\varepsilon,\delta}^{n+j-1}\int_{\Lambda^{n+j}} \psi(\vc z_{j})\prod_{\gamma, \gamma' \in E(G)} \eta^{\varepsilon,\alpha}(x_\gamma,x_{\gamma'}) \mathrm d\rho^{\otimes(n+j)}(\vc z_{j+n}), \qquad j \in\{1,2\}
        \]
        where \(E \in \mathcal{C}_{n+j}\). For \(s<d\), we expand the Mayer function \eqref{eq:mayer} and, for \(x \neq y\), we get
        \[
        \mu_{\varepsilon,\delta}\, \eta^{\varepsilon,\alpha}(x,y)= -\beta|x-y|^{-s} + O(\varepsilon^{\alpha+s}).
        \]
        This is sufficient to prove that the only non-vanishing terms come from trees. For \(s>d\), general graphs with cycles could contribute. We pass to the limit by dominated convergence.
        
        Since the higher-order cumulants of the fluctuation field are vanishing,
        we conclude that the limit field process is Gaussian with limit covariance structure given by \eqref{eq:limitcov} or \eqref{eq:limitcov2}. 
     \end{proof}

\section{Proof of Proposition \ref{graph_red}: Graph reduction procedure}
\label{sec4}

  In this section, we prove the key proposition to get pointwise estimates on truncated correlation functions for long-range systems. 
    The proof strategy is to estimate the tree--integrals in \eqref{int_graph}
    by integrals over simpler trees algorithmically generated from the original ones. 
    We shall see that the number of resulting trees grows exponentially in terms of \(j\), and they are of a form that can be more easily analyzed.
    The section is organized as follows:
    \begin{enumerate}
        \item In Subsection \ref{sec4.1}, we introduce the necessary concepts and notation.
        \item In Subsection \ref{subsect:reduction-to-cc}, we decompose a given tree graph into connected components. Each component consists either of connected free vertices, i.e. vertices that are integrated over in \eqref{int_graph}, or of connected fixed vertices, i.e. the given points appearing as arguments in the truncated correlation functions.
        \item In Subsection \ref{sec4.3}, we construct reduced connected components by integrating out all leaves and chains inside the connected components.
        \item In Subsection \ref{sec4.4}, a specific operation is introduced to simplify the boundary structure of the reduced components. As a result, each free vertex within a component will be connected to at most one fixed boundary vertex.
        \item In Subsection \ref{subsect:final-processing-algorithm}, we introduce a final algorithm. This algorithm is applied to the modified reduced component of the previous step and it produces trees in which each vertex is connected to at most two other points. This property permits a straightforward integration of the resulting tree by Fubini's theorem.
        \item In Subsection \ref{sec4.6}, we conclude the proof of Proposition \ref{graph_red}.
    \end{enumerate}

\subsection{Notation}
\label{sec4.1}
    Given $j \in \N_1, n \in \N_0$, we define \(\mathcal{T}_{j,n}\) 
    as the set of $2$--color trees \(T = (V(T), E(T), c_{j,n})\),
    with the vertices \(V(T)  = \{1, \dots, j+n\}\) and the edges
    \(E(T) \subset \{\{\gamma,\gamma'\} \colon \gamma \in V(T), \gamma' \in V(T)\}\) satisfying the tree property, and the vertex color map \(c_{j,n} \colon V(T) \to \{\bullet, \circ\}\) defined as
    \begin{align}
        \label{eqdef:color-map}
        c_{j,n}(\gamma) \coloneqq \begin{cases}
                \bullet, \quad &\gamma = 1, \dots, j, \\
                \circ, \quad &\gamma = j+1, \dots, j+n.
                \end{cases}
    \end{align}
    For a \(2\)--color tree \(T \in \mathcal{T}_{j,n}\)
    and a point \(\vc{x}_j=(x_1,\dots,x_j)\in \Lambda^j\), 
    we use the following notation for the corresponding graph integral:
    \begin{align}
    \label{eq:app:mainint}
        \inI{T}(\vc{x}_j) \equiv \inI{T}^{\varepsilon, \alpha}(\vc{x}_j) \coloneqq
        \int_{\Lambda^{n}} \prod_{\{\gamma, \gamma'\}\in E(T)} \abs{\eta^{\varepsilon,\alpha}(x_{\gamma},x_{\gamma'})} \prod_{\ell=j+1}^{n+j} \rmd \rho(x_{\ell}),
    \end{align}
    with the intention that
    \(c_{j,n}(\gamma) = \bullet\) if \(\gamma = 1,\dots, j\), i.e. $x_\gamma$ is one of the fixed variables in \eqref{eq:app:mainint}, 
    while \(c_{j,n}(\gamma) = \circ\) if $\gamma = j+1,\dots, n$, i.e. $x_\gamma$ is one of the variables that are being integrated over in the integral \eqref{eq:app:mainint}. 

   With the notation \(C \subset T\), we mean that a graph \(C\) is an \emph{induced subgraph} of \(T\).  If \(C \subset T\) and \(\gamma \in V(C)\), we define
    \[\deg_{C}(\gamma) \coloneqq \abs{\{\gamma' \in V(C) : \{\gamma, \gamma'\} \in E(C)\} }\]
    to be the degree of \(\gamma\) in \(C\). Similarly, we write \(|C|\coloneqq |V(C)|\) for the \emph{order} of \(C\), namely the number of vertices of \(C\).

   We denote by \(\free{F}(T)\) the induced subgraph of \(T\) generated by the vertices $\gamma \in V(T)$ such that $c_{j,n}(\gamma)=\circ$,
   and by \(\fix{F}(T)\) the induced subgraph generated by the vertices  $\gamma \in V(T)$ such that $c_{j,n}(\gamma)=\bullet$.
   Since \(T\) is a tree, both \(\free{F}(T)\) and \(\fix{F}(T)\) are forests.
   For a forest \(F\), we denote by \(\mathcal{C}(F)\)
   the set of \emph{connected components} of \(F\), i.e. the set of maximal connected induced subgraphs of \(F\). In particular, since \(F\) is a forest, every element of \(\mathcal{C}(F)\) is a tree.

    If \(C \subset T\) satisfies \(c_{j,n}(\gamma) = \circ\) for all \(\gamma \in V(C)\), 
    i.e. \(C \subset \free{F}(T)\), we define the \emph{fixed boundary} of \(C\) as the set
    \begin{align*}
        \fix{\partial}C &\coloneqq 
        \{\gamma \in V(T) : c_{j,n}(\gamma) = \bullet \text{ and there exists } \gamma' \in V(C)\text{ s.t. } \{\gamma, \gamma'\} \in E(T)\},
    \end{align*}
    and \(\overline{C}\), the \emph{closure} of \(C\), as the induced subgraph generated by \(V(C) \cup \fix{\partial}C\).
    Note that \(\fix{\partial}C \subset V(\fix{F}(T))\). Finally, by a decomposition of \(T\) we mean a family of subgraphs of \(T\) whose graph union is \(T\).

    \subsection{Reduction to connected components}
    \label{subsect:reduction-to-cc}
    For a given \(T \in \mathcal{T}_{j,n}\),
    we first introduce the following decomposition:
		\begin{align*}
		T = \bigcup_{i=1}^{\ell_1} \overline{C_{i}} \cup \bigcup_{i=1}^{\ell_2} \overline{L_{i}} \cup \fix{F}(T),
		\end{align*}
    with $C_i, L_i \in \mathcal{C}(\free{F}(T))$, see Figure \ref{fig:tree1}.
    The connected components \(L_{i}\) are those that satisfy \(\abs{\fix{\partial} L_i} = 1\),
	whereas the connected components \(C_{i}\) satisfy \(\abs{\fix{\partial}C_i} \geq 2\).
        
        \begin{lemma}[Component number bound]
            \label{lemma:component-number-bound}
            The
		number \(\ell_1\) of the connected components \(\{C_i\}_{i=1}^{\ell_1}\) is bounded like \(\ell_1 \leq j-1\).
        \end{lemma}
        \begin{proof}
           Since \(T\) is a tree, we have 
           \begin{align}
           \label{ineq:component-boundary-increase}
           \abs{\fix{\partial}C_i \setminus \bigcup_{k=1}^{i-1} \fix{\partial}C_k} \geq 1.
           \end{align}
           for all \(i = 1, \dots, \ell_1\). Indeed, if this was not true, then there would be a cycle in \(T\).
           
           Let us write the union of the sets \(\fix{\partial}C_i\) in terms of a family of auxiliary sets \(B_i\):
           \begin{align*}
               \bigcup_{i=1}^{\ell_1}\fix{\partial}C_i = \bigcup_{i=1}^{\ell_1} B_{i},
            \end{align*}
            with the sets \(B_{i}\) defined as
            \begin{align*}
            B_{i} \coloneqq \begin{cases}
                \fix{\partial}C_1, \quad &i = 1 \\
                \fix{\partial}C_i \setminus \left(\bigcup_{j=1}^{i-1} \fix{\partial}C_{j}\right), \quad &i=2, \dots, \ell_1.
            \end{cases}
           \end{align*}
           
            Since the sets \(B_{i}\) satisfy \(B_{i+1} \cap B_{i} = \emptyset\) for all \(i = 1, \dots, \ell_1-1\), we can use countable additivity of the counting measure together with \eqref{ineq:component-boundary-increase} to get the following lower bound on the number of boundary points of the connected components \(C_i\):
            \begin{align}
            \label{ineq:lower-bound-component-boundary}
            \abs{\bigcup_{i=1}^{\ell_1}\fix{\partial}C_i} = \abs{\bigcup_{i=1}^{\ell_1}B_i} = \sum_{i=1}^{\ell_1} \abs{B_i} \geq \ell_1 + 1.
            \end{align}

            However, because \(\bigcup_{i=1}^{\ell_1} \fix{\partial}C_i \subset V(\fix{T})\), and since the number of fixed vertices in the tree satisfies \(\abs{\fix{T}} = j\), it follows from \eqref{ineq:lower-bound-component-boundary} that \(j \geq \ell_1 + 1\).
        \end{proof}

    With this decomposition at our disposal, we can use it together with Fubini's theorem, to express the integral \(\inI{T}(\vc{x}_j)\) in the following way
    \begin{align*}
        \inI{T}(\vc{x}_j) = \left(\prod_{\{\gamma_1, \gamma_2\} \in E(\fix{F}(T))} \abs{\eta^{\varepsilon,\alpha}(x_{\gamma_1}, x_{\gamma_2})}\right) \left(\prod_{i=1}^{\ell_1} \inI{\overline{C_i}}(\vc{x}_{\fix{\partial}C_i})\right)\left(\prod_{i=1}^{\ell_2}\inI{\overline{L_i}}(\vc{x}_{\fix{\partial}L_i})\right).
    \end{align*}
    Here the integral corresponding to each tree \(\overline{C_i}\) is given by
    \begin{align}
    \label{eq_int_C_i}
        \inI{\overline{C_i}}(\vc{x}_{\fix{\partial}C_i}) = \int_{\Lambda^{\abs{C_i}}} \left(\prod_{\{\gamma_1, \gamma_2\} \in E\left(\overline{C_i}\right) } \abs{\eta^{\varepsilon,\alpha}(x_{\gamma_1}, x_{\gamma_2})}\right) \rmd \rho^{\otimes \abs{V(C_i)}}(\vc{x}_{V(C_i)}),
    \end{align}
    and similarly for the trees \(\overline{L_i}\)
    \begin{align}
    \label{eq:int_L_i}
        \inI{\overline{L_i}}(\vc{x}_{\fix{\partial}L_i}) = \int_{\Lambda^{\abs{L_i}}} \left(\prod_{\{\gamma_1, \gamma_2\} \in E\left(\overline{L_i}\right)} \abs{\eta^{\varepsilon,\alpha}(x_{\gamma_1}, x_{\gamma_2})}\right) \rmd \rho^{\otimes \abs{V(L_i)}}(\vc{x}_{V(L_i)}).
    \end{align}
    
    \begin{lemma}
        Let \(\inI{\overline{L_i}}\) be the integral given in \eqref{eq:int_L_i}. The following bound holds
        \begin{align}
            \label{ineq:linear-component}
            \inI{\overline{L_i}}(x) \leq \bar \rho^{\abs{L_i}}(C_{\beta,s} \omega(\varepsilon))^{\abs{L_i}}, \quad x =\vc{x}_{\fix{\partial}L_i} \in \Lambda
        \end{align}
        where \(\abs{L_i}\) stands for the number of vertices in \(L_{i}\) and \(\omega(\varepsilon)\) is defined in \eqref{eq:omega_epsilon}.
    \end{lemma}
    \begin{proof}
        By definition, there is only one fixed node in \(\overline{L_i}\). Therefore, the integral \(\inI{\overline{L_i}}\) in \eqref{eq:int_L_i} can be solved iteratively by Fubini's theorem, where each integration over a free variable \(x_{\gamma}\), \(\gamma \in L_i\), yields
        \begin{equation}
        \label{eq:int_l}
        \int_\Lambda \abs{\eta^{\varepsilon,\alpha}(x,x_\gamma)} \rmd \rho(x_\gamma) \le C_{\beta,s} \,\bar \rho\, \omega(\varepsilon).
        \end{equation}
        In total, there are \(\abs{L_i}\) free variables, whereby \eqref{ineq:linear-component} is verified.
    \end{proof}

    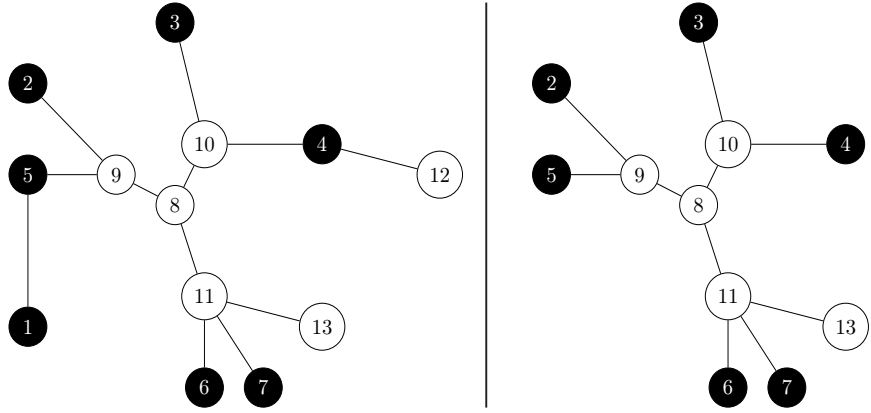
\begin{figure}[h]
    \centering
      \begin{subfigure}
        \centering
            \resizebox{30em}{14em}{
\begin{tikzpicture}
    \node[shape=circle,draw=black, fill=black, text=white] (1) at (0,-1) {1};
    \node[shape=circle,draw=black, fill=black, text=white] (2) at (0,3) {2};
    \node[shape=circle,draw=black, fill=black,text=white] (3) at (2.5,4) {3};
    \node[shape=circle,draw=black, fill=black,text=white] (4) at (5,2) {4};
    \node[shape=circle,draw=black, fill=black,text=white] (5) at (0,1.5) {5};
    \node[shape=circle,draw=black, fill=black,text=white] (6) at (3,-2) {6};
    \node[shape=circle,draw=black, fill=black,text=white] (7) at (4,-2) {7};
    
    \node[shape=circle,draw=black] (8) at (2.5,1) {8};
    \node[shape=circle,draw=black] (9) at (1.5,1.5) {9};
    \node[shape=circle,draw=black] (10) at (3,2) {10};
    \node[shape=circle,draw=black] (11) at (3,-0.5) {11};
    \node[shape=circle,draw=black] (12) at (7,1.5) {12};
    \node[shape=circle,draw=black] (13) at (5,-1) {13};

    \path [-] (1) edge node[left] {} 
    (5);
    \path [-] (5) edge node[above] {} (9);
    \path [-] (2) edge node[above] {} (9);
    \path [-] (8) edge node[above] {} (9);
    \path [-] (6) edge node[above] {} (11);
    \path [-] (7) edge node[above] {} (11);
    \path [-] (10) edge node[above] {} (3);
    \path [-] (8) edge node[above] {} (10);
    \path [-] (4) edge node[above] {} (10);
    \path [-] (4) edge node[above] {} (12);
    \path [-] (8) edge node[above] {} (11);
    \path [-] (11) edge node[above] {} (13);
\end{tikzpicture}
\,\, \vline\vline
\qquad 
\begin{tikzpicture}
    \node[shape=circle,draw=black, fill=black, text=white] (2) at (0,3) {2};
    \node[shape=circle,draw=black, fill=black,text=white] (3) at (2.5,4) {3};
    \node[shape=circle,draw=black, fill=black,text=white] (4) at (5,2) {4};
    \node[shape=circle,draw=black, fill=black,text=white] (5) at (0,1.5) {5};
    \node[shape=circle,draw=black, fill=black,text=white] (6) at (3,-2) {6};
    \node[shape=circle,draw=black, fill=black,text=white] (7) at (4,-2) {7};
    
    \node[shape=circle,draw=black] (8) at (2.5,1) {8};
    \node[shape=circle,draw=black] (9) at (1.5,1.5) {9};
    \node[shape=circle,draw=black] (10) at (3,2) {10};
    \node[shape=circle,draw=black] (11) at (3,-0.5) {11};
    \node[shape=circle,draw=black] (13) at (5,-1) {13};

    \path [-] (5) edge node[above] {} (9);
    \path [-] (2) edge node[above] {} (9);
    \path [-] (8) edge node[above] {} (9);
    \path [-] (6) edge node[above] {} (11);
    \path [-] (7) edge node[above] {} (11);
    \path [-] (10) edge node[above] {} (3);
    \path [-] (8) edge node[above] {} (10);
    \path [-] (4) edge node[above] {} (10);
    \path [-] (8) edge node[above] {} (11);
    \path [-] (11) edge node[above] {} (13);
\end{tikzpicture}}
    \end{subfigure}
         \caption{Left: Example of a tree arising in the cluster expansion. Right: One of the connected components of the free vertices of the previous tree together with its fixed boundary.}
                     \label{fig:tree1}
    \end{figure}
    
    \subsection{From connected components to reduced components}
    \label{sec4.3}
    Fix now \(i\in \{1,\dots, \ell_1\}\) and consider the connected component \(C_i\).
    In order to estimate the integral \eqref{eq_int_C_i}, 
    we begin by performing two elementary reductions on the tree $C_i$, 
    corresponding to partially integrating \eqref{eq_int_C_i} over simple substructures of $C_i$. In particular:
    \begin{itemize}
    \item[--] We first iteratively integrate out all vertices \(\gamma \in V(C_i)\) such that \(\deg_{C_i}(\gamma) = 1\). 
    This is done by repeatedly applying Fubini's theorem and the estimate \eqref{eq:int_l} already used, see Figure \ref{fig:tree3}.
    We denote by \(C'_i\) the resulting connected component, obtained after all degree-one vertices have been iteratively removed.

    \item[--]  We then integrate out all chains, i.e. variables corresponding to vertices \(\gamma \in C_i'\) such that \(\deg_{\overline{C_i}}(\gamma) = 2\).
    This is done by using again Fubini's theorem and the upper bound
    \begin{align}
    \label{eq:int_c}
    \qquad
    \int_{\Lambda}
       \abs{\eta^{\varepsilon,\alpha}(x_{\gamma}, x_{\gamma'})}\abs{\eta^{\varepsilon,\alpha}(x_{\gamma'}, x_{\gamma''})}\rmd \rho(x_{\gamma'}) \leq \bar \rho\, C_{\beta,s}\, \omega(\varepsilon) \abs{\eta^{\varepsilon,\alpha}(x_{\gamma},x_{\gamma''})}.
    \end{align}
    \end{itemize}
    By iterating these two reduction steps, we have thus proved the following lemma.
    
    \begin{lemma}
    For all \(\vc{x} = ({x_{1}, \dots, x_{\abs{\fix{\partial}C_i}}})\), we have
    \begin{align*}
        \inI{\overline{C_i}}(\vc{x}) \leq \left(\bar \rho\, C_{\beta, s}\,\omega(\varepsilon)\right)^{\mathfrak{l}\left(\overline{C_i}\right) + \mathfrak{c}\left(\overline{C_i'}\right)} \inI{\overline{\reduc}}(\vc{x}).
    \end{align*}
    where \(\reduc\) is the reduced connected component where all leaves and chains have been removed and \(\mathfrak{l}(\overline{C_i})\) and \(\mathfrak{c}(\overline{C_i'})\) are the total number of leaves and chains eliminated in the process. Notice that \(\fix{\partial}\reduc = \fix{\partial}C_{i}\). 
    \end{lemma}

    The next lemma allows us to control the order of the reduced components in terms of their fixed boundaries.
    \begin{lemma}
        The number of vertices in the reduced component \(\reduc\) satisfies the bound
        \begin{align*}
            \abs{\reduc} \leq \abs{\fix{\partial}C_i}-2.
        \end{align*}
    \end{lemma}
    \begin{proof}
        Since \(\overline{\reduc}\) is a tree, 
        the number of edges in it is \(|\overline{\reduc}|-1 = |\reduc|+|\fix{\partial}\reduc|-1\). Hence
        \begin{align}
            \label{eq:double-count-edges}
        \sum_{\gamma \in \overline{\reduc}} \deg_{\overline{\reduc}}(\gamma)=2(|\reduc|+|\fix{\partial} \reduc|-1),
        \end{align}
        since each edge gets counted twice in the above sum.
        Moreover, each \(\fix{\gamma} \in \fix{\partial} \reduc\) has degree \(1\), while each \(\gamma' \in \reduc\) has degree at least \(3\) in \(\overline{\reduc}\), since all internal leaves and chains have been removed. Therefore,
         \begin{align}
            \label{ineq:internal-vertex-degree-lower-bound}
        \sum_{\gamma \in \overline{\reduc}} \deg_{\overline{\reduc}}(\gamma)\geq 3|\reduc|+|\fix{\partial}\reduc|.
        \end{align}
        Since the reduction procedure does not change the boundary, by combining \eqref{eq:double-count-edges} and \eqref{ineq:internal-vertex-degree-lower-bound} we get
        \[
        3|\reduc|+|\fix{\partial}C_i| \leq 2(|\reduc|+|\fix{\partial}C_i|-1),
        \]
        which proves the claim.
    \end{proof}
    \begin{figure}
        \centering
        \begin{tikzpicture}
    \node[shape=circle,draw=black, fill=black, text=white] (2) at (0,3) {2};
    \node[shape=circle,draw=black, fill=black,text=white] (3) at (2.5,4) {3};
    \node[shape=circle,draw=black, fill=black,text=white] (4) at (5,2) {4};
    \node[shape=circle,draw=black, fill=black,text=white] (5) at (0,1.5) {5};
    \node[shape=circle,draw=black, fill=black,text=white] (6) at (3,-2) {6};
    \node[shape=circle,draw=black, fill=black,text=white] (7) at (4,-2) {7};
    
    \node[shape=circle,draw=black] (8) at (2.5,1) {8};
    \node[shape=circle,draw=black] (9) at (1.5,1.5) {9};
    \node[shape=circle,draw=black] (10) at (3,2) {10};
    \node[shape=circle,draw=black] (11) at (3,-0.5) {11};

    \path [-] (5) edge node[above] {} (9);
    \path [-] (2) edge node[above] {} (9);
    \path [-] (8) edge node[above] {} (9);
    \path [-] (6) edge node[above] {} (11);
    \path [-] (7) edge node[above] {} (11);
    \path [-] (10) edge node[above] {} (3);
    \path [-] (8) edge node[above] {} (10);
    \path [-] (4) edge node[above] {} (10);
    \path [-] (8) edge node[above] {} (11);
\end{tikzpicture}
        \caption{The tree obtained from the connected component tree in Figure \ref{fig:tree1} by removing chains and leaves.}
        \label{fig:tree3}
    \end{figure}
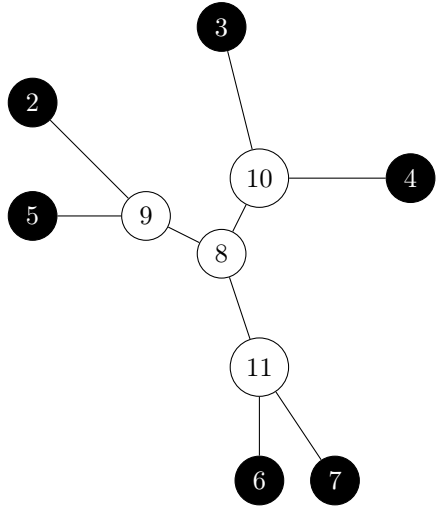
    
    \subsection{From reduced components to modified reduced components}
    \label{sec4.4}

    Recall that, for each reduced connected component \(\reduc\), 
    the integral associated with its closure \(\overline{\reduc}\) is given by
    \begin{align*}
        \inI{\overline{\reduc}}(\vc{x}_{\fix{\partial}C_i})
        =
         \int_{\Lambda^{|\reduc|}} \left(\prod_{\{\gamma_1, \gamma_2\} \in E\left(\overline{\reduc}\right)} \abs{\eta^{\varepsilon,\alpha}(x_{\gamma_1}, x_{\gamma_2})}\right) \rmd \rho^{\otimes \abs{V(\reduc)}}(\vc{x}_{V(\reduc)}).
    \end{align*}
    For \(\lambda \in (0,1)\), define the modified Mayer function 
        \begin{align}
            \label{eqdef:zeta-half}
            \eta_{1;\lambda}^{\varepsilon,\alpha}(x_{\gamma_1},x_{\gamma_2}) \coloneqq
            \begin{cases}
            \eta^{\varepsilon,\alpha}_\lambda(x_{\gamma_1},x_{\gamma_2}), \quad &\text{ if exactly one of } \gamma_1, \gamma_2 \text{ belongs to }  V(\reduc) \\
            \eta^{\varepsilon,\alpha}(x_{\gamma_1},x_{\gamma_2}), \quad &\text{ otherwise},
            \end{cases}
        \end{align}
    where 
    \begin{align}
    \label{defeq:eta_dil}
    \eta_\lambda^{\varepsilon, \alpha}(x,y)\coloneqq\eta^{\varepsilon,\alpha}(\lambda x, \lambda y).
    \end{align}
    \begin{lemma}
    \label{lemma_mod_conn}
        There exists a finite family \(\mathcal{D} = \mathcal{D}(\vc{x}_{\fix{\partial}C_i},\overline{\reduc})\) of subsets of \(\Lambda^{\abs{\reduc}}\) 
        forming, up to null sets, a partition of \(\Lambda^{\abs{\reduc}}\) 
        such that
        \begin{align}
            \label{ineq:domain-decomposition-upper-bound}
             \inI{\overline{\reduc}}(\vc{x}_{\fix{\partial}C_i}) \leq \sum_{D \in \mathcal{D}} \inI{T\left(\overline{\reduc};D\right)}^{D}
        \end{align}
        The integrals on the right-hand side of \eqref{ineq:domain-decomposition-upper-bound} are defined by
        \begin{align}
            \label{ineq:tree-integral-subdomain}
            \inI{T\left(\overline{\reduc};D\right)}^D &\equiv \inI{T \left(\overline{\reduc};D\right)}^{D}(\vc{x}_{\fix \partial C_i})\\
          &\coloneqq \int_{\Lambda^{\abs{\reduc}}} 
            \cf{\vc{x}_{V(\reduc)} \in D} \left(\prod_{\{\gamma_1, \gamma_2\} \in E\left(T\left(\overline{\reduc}; D\right)\right)} \abs{\eta_{1;1/2}^{\varepsilon,\alpha}(x_{\gamma_1}, x_{\gamma_2})}\right)
            \rmd \rho^{\otimes \abs{V(\reduc)}}(\vc{x}_{V(\reduc)}),\nonumber
        \end{align}
        where the modified Mayer function \(\eta^{\varepsilon,\alpha}_{1;1/2}\) is defined in \eqref{eqdef:zeta-half}.
        Given a configuration \(\vc{x}_{\reduc} = (x_{\gamma})_{\gamma \in V(\reduc)}\) belonging to \(D \in \mathcal{D}\), 
        the tree \(T\left(\overline{\reduc}; D\right)\) is obtained as follows:
        for each \(\gamma \in V(\reduc)\), let \(\gamma_*\) its closest fixed boundary neighbor.
        Then every edge connecting \(\gamma\) to a fixed boundary vertex \emph{other than} \(\gamma_*\)
        is replaced by an edge connecting that boundary vertex to \(\gamma_*\), see Figure \ref{fig:voronoi1}.

         The cardinality of \(\mathcal{D}\) is controlled exponentially by the cardinality of \(\fix \partial \reduc = \fix \partial C_{i}\),  to wit:
        \begin{align}
        \label{eq:Dbound}
        |\mathcal{D}| \le \exp\left\{\frac{|\fix \partial C_i|}{\rme}\right\}.
        \end{align}
    \end{lemma}
       \begin{figure}
        \centering
        \begin{tikzpicture}
    \node[shape=circle,draw=black, fill=black, text=white] (1) at (0,3) {1};
    \node[shape=circle,draw=black, fill=black,text=white] (2) at (2.5,4) {2};
    \node[shape=circle,draw=black, fill=black,text=white] (3) at (5,2) {3};

    \node[shape=circle, draw = black, fill = white, text = black] (4) at (3,0) {4};

    \draw[dashed] (0.25, 6.0) to (2.25,1);

    \draw[dashed] (2.15,1.0) to (5.75,5.5);

    \draw[dashed] (2.0, 0.0) to (2.2, 1.0);
    

    \node (a) at (1.1,1.1) {$d_{14}$};
    \node (b) at (2.4,2.5) {$d_{24}$};
    \node (c) at (4.5, 0.9) {$d_{34}$};
   \draw (1) to (4);
   \draw (2) to (4);
   \draw (3) to (4);
\end{tikzpicture}
\,\, \vline\vline
\qquad 
\begin{tikzpicture}
    \node[shape=circle,draw=black, fill=black, text=white] (1) at (0,3) {1};
    \node[shape=circle,draw=black, fill=black,text=white] (2) at (2.5,4) {2};
    \node[shape=circle,draw=black, fill=black,text=white] (3) at (5,2) {3};

    \node[shape=circle, draw = black, fill = white, text = black] (4) at (3,0) {4};

    \draw[dashed] (0.25, 6.0) to (2.25,1);

    \draw[dashed] (2.15,1.0) to (5.75,5.5);

    \draw[dashed] (2.0, 0.0) to (2.2, 1.0);
    

    \node (a) at (3.0,2.8) {$\frac{1}{2}d_{13}$};
    \node (b) at (3.6,3.8) {$\frac{1}{2}d_{23}$};
    \node (c) at (4.5, 0.9) {$d_{34}$};
   \draw (1) to (3);
   \draw (2) to (3);
   \draw (3) to (4);
\end{tikzpicture}
        \caption{Left: Vertex \(4\) corresponding to the only free variable in this tree graph has Vertex \(3\) as its closest fixed neighbor. Right: The bound from above obtained by replacing the original tree with one in which Vertex 4 has only one fixed neighbor and every other previous neighborhood of Vertex 4 is connected to that fixed neighbor.}
        \label{fig:voronoi1}
    \end{figure}
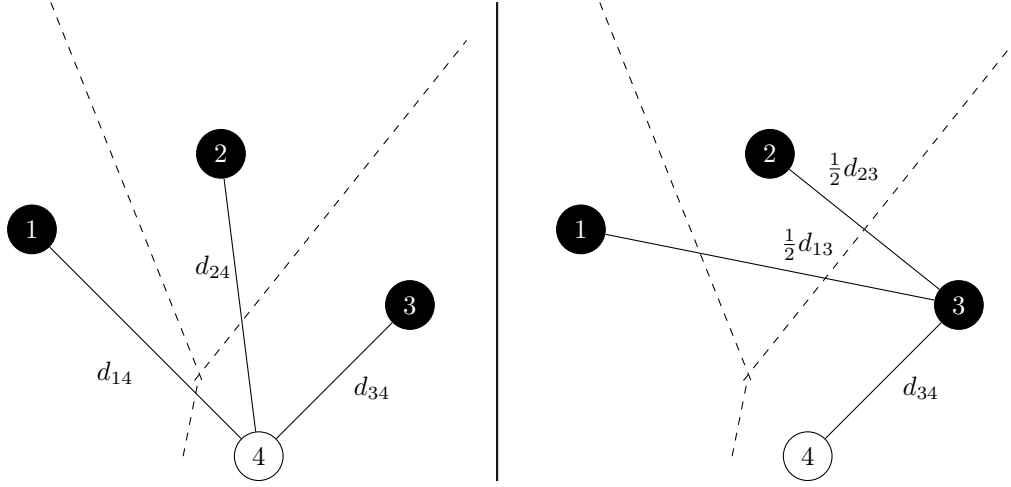
    \begin{proof}
       Let \({\vc x}_{\reduc}=(x_{\gamma})_{\gamma \in \reduc} \in \Lambda^{\abs{\reduc}}\).
       For every \(\gamma \in V(\reduc)\), consider the set of fixed neighbor vertices of \(\gamma\):
        \begin{align*}
            \fix{\partial}_\gamma \reduc \coloneqq \left\{\gamma' \in \fix{\partial}C_i \colon \{\gamma, \gamma'\} \in E\left(\overline{\reduc}\right)\right\} \subset \fix{\partial}\reduc.
        \end{align*}
        Let \(\gamma_{\ast} \equiv \gamma_{\ast}(x_{\gamma}) \in \argmin\{\abs{x_{\gamma}- x_{\gamma'}} : \gamma' \in \fix{\partial}_{\gamma}\reduc\}\).
        Then, for any \(\gamma' \in \fix{\partial}_\gamma \reduc \setminus \{\gamma_\ast\}\), we have 
        \[
        \abs{x_{\gamma_\ast}-x_{\gamma'}} \leq 
        \abs{x_{\gamma_\ast}-x_\gamma}+
        \abs{x_\gamma - x_{\gamma'}}
        \leq 2\abs{x_{\gamma}-x_{\gamma'}},
        \]
        and by the monotonicity of \(\abs{\eta^{\varepsilon}(x_{\gamma}, x_{\gamma'})}\) in the distance \(\abs{x_{\gamma}-x_{\gamma'}}\), the following bound holds 
        \begin{align*}
            \abs{\eta^{\varepsilon,\alpha}(x_{\gamma}, x_{\gamma'})} \leq \abs{\eta_{1/2}^{\varepsilon,\alpha}(x_{\gamma'}, x_{\gamma_\ast})}.
        \end{align*}
        Hence every edge connecting \(\gamma\) to a fixed boundary vertex \(\gamma' \in \fix \partial_\gamma \reduc \setminus \{\gamma_\ast\}\)
        may be replaced by an edge \(\{\gamma_\ast,\gamma'\}\)
        with a loss in the distance factor \(1/2\), yielding a change in the Mayer--function \(\eta^{\varepsilon,\alpha}(x_{\gamma}, x_{\gamma'}) \mapsto \eta^{\varepsilon,\alpha}_{1/2}(x_{\gamma'}, x_{\gamma_*})\).

        Applying this procedure simultaneously to all vertices \(\gamma \in V(\reduc)\) defines a modified tree \(\mathfrak{t}(\vc{x}_{\reduc})\), and, outside of exceptional sets of zero Lebesgue measure, each tree is uniquely defined by \(\vc{x}_{\reduc}\).
        It should be noted, however, that different configurations \(\vc{x}_{\reduc}\in \Lambda^{|\reduc|}\) may lead to different modified trees. 
        We denote by \(\mathcal{T}^*\) the collection of all modified trees that can be obtained by this procedure, with the further restriction that the tree is uniquely defined. We can do this, since  the Lebesgue measure of points failing this uniqueness condition is zero.

        For each \(\mathfrak{T}\in \mathcal{T}^*\), we define the following subsets of \(\Lambda^{\abs{\reduc}}\)
        \begin{align*}
       D\equiv D_{\mathfrak{T}} \coloneqq \{\vc{x}_{\reduc}\in \Lambda^{|\reduc|} : \mathfrak{t}(\vc{x}_{\reduc})=\mathfrak{T}\},
       \end{align*}
       and then set \(\mathcal{D}\coloneqq\{D_{\mathfrak{T}} : \mathfrak{T} \in \mathcal{T}^*\}\).
       By construction, the sets in \(\mathcal{D}\) form, up to modifications on a null set, a partition of \(\Lambda^{|\reduc|}\). Moreover,
       \begin{align*}
       \abs{\mathcal{D}} = \abs{\mathcal{T}^*}.
       \end{align*}
       Integrating over each domain and summing over the trees \(\mathfrak{T}\in \mathcal{T}^*\) yields \eqref{ineq:domain-decomposition-upper-bound}.    

       It remains to estimate the cardinality of the family \(\mathcal{D}\). For every \(\gamma \in V(\reduc)\), the construction of the modified tree
       requires choosing one vertex in \(\fix \partial_\gamma \reduc\). Therefore
       \begin{align}
            \label{eq:number-of-tree-families}
       |\mathcal{D}|\le \prod_{\gamma \in V(\reduc)}
       p_\gamma, \qquad
       p_\gamma\coloneqq
      \abs{\fix{\partial}_\gamma \reduc},
        \end{align}
        with the constraint that \(\sum_{\gamma \in V(\reduc)}p_\gamma = \abs{\fix{\partial}C_i}\). Indeed, each fixed boundary vertex is connected to exactly one free vertex in \(\reduc\).
        
        By the \emph{arithmetic mean, geometric mean} (AM--GM) inequality,
        the product in \eqref{eq:number-of-tree-families} is maximized when \(p_\gamma = \frac{\abs{\fix{\partial}C_i}}{\abs{\reduc}}\) for all \(\gamma \in V(\reduc)\). Hence
        \begin{align}
            \label{ineq:number-of-tree-domains}
            |\mathcal{D}|\leq \left(\frac{\abs{\fix{\partial}C_i}}{\abs{\reduc}}\right)^{\abs{\reduc}}.
        \end{align}
        Maximizing the right hand side of \eqref{ineq:number-of-tree-domains} over \(|\reduc|\) subject to the condition
        \(
            1 \leq  \abs{\reduc} \leq \abs{\fix{\partial}C_i}
        \)
        yields the upper bound in \eqref{eq:Dbound}.
    \end{proof}

    \subsection{Algorithm inside the modified reduced component}
    \label{subsect:final-processing-algorithm}

    Our next task is to estimate the integrals \eqref{ineq:tree-integral-subdomain}, namely
    \begin{align}
    \label{sec4.6:int}
        \inI{T\left(\overline{\reduc};D\right)}^D = 
        \int_{\Lambda^{\abs{\reduc}}} \cf{\vc{x}_{V(\reduc)}\in D} \left(\prod_{\{\gamma_1,\gamma_2\} \in E\left(T\left(\overline{\reduc};D\right)\right)} \abs{\eta_{1;1/2}^{\varepsilon,\alpha}(x_{\gamma_1}, x_{\gamma_2})}\right)\rmd \rho^{\otimes \abs{V(\reduc)}}(\vc{x}_{V(\reduc)}).
    \end{align}

    Given \(D \in \mathcal{D}\), let \(\mathfrak{T} \coloneqq T\left(\overline{\reduc}; D\right)\). 
    To estimate \eqref{sec4.6:int}, we apply the following Algorithm \ref{alg:algorithm-final-step}. 
    The inputs of the algorithm are the tree \(\mathfrak{T}\),
    the fixed vertices \(\vc{x}_{\fix\partial{\mathfrak{T}}}\) and \(\vc{x}_{V(\reduc)} \in D \in \mathcal{D}\).
    It produces a reduced weighted tree \(\finaltr\) that depends on all the above objects. Final trees \(\finaltr\) that can be produced by the algorithm from \(\mathfrak{T}\), \(D\), and \(\vc{x}_{\fix \partial \mathfrak{T}}\) by considering different input points \(\vc{x}_{V(\reduc)} \in D\) are said to be \emph{accessible} from \(\mathfrak{T}\) and denote this relation by \(\finaltr \succ \mathfrak{T}\).
    For a fixed tree \(\mathfrak{T}\), fixed \(\vc{x}_{\fix\partial{\mathfrak{T}}}\) and \(D \in \mathcal{D}\), the total number of distinct accessible trees is bounded by \(4^{\abs{\reduc}}\). 
    
    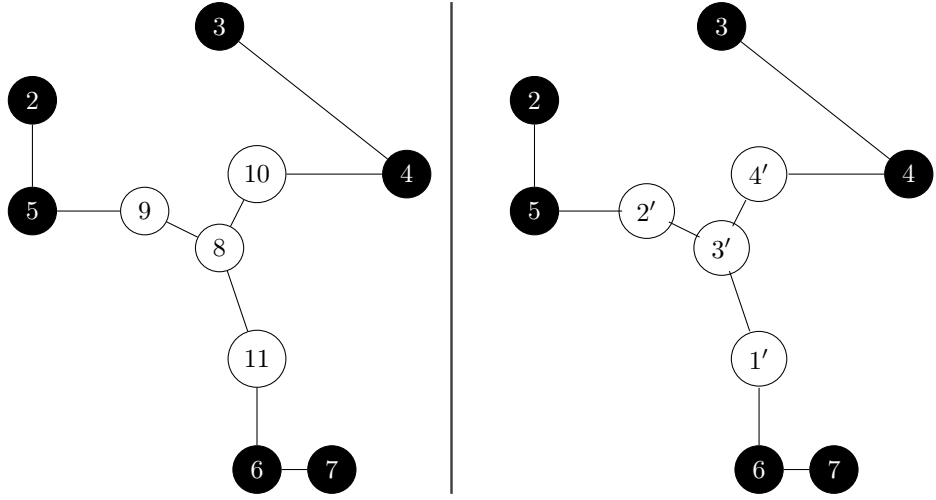
\begin{figure}
        \centering
        \resizebox{32em}{17em}{\begin{tikzpicture}
    \node[shape=circle,draw=black, fill=black, text=white] (2) at (0,3) {2};
    \node[shape=circle,draw=black, fill=black,text=white] (3) at (2.5,4) {3};
    \node[shape=circle,draw=black, fill=black,text=white] (4) at (5,2) {4};
    \node[shape=circle,draw=black, fill=black,text=white] (5) at (0,1.5) {5};
    \node[shape=circle,draw=black, fill=black,text=white] (6) at (3,-2) {6};
    \node[shape=circle,draw=black, fill=black,text=white] (7) at (4,-2) {7};
    
    \node[shape=circle,draw=black] (8) at (2.5,1) {8};
    \node[shape=circle,draw=black] (9) at (1.5,1.5) {9};
    \node[shape=circle,draw=black] (10) at (3,2) {10};
    \node[shape=circle,draw=black] (11) at (3,-0.5) {11};

    \path [-] (5) edge node[above] {} (9);
    \path [-] (2) edge node[left] {} (5);
    \path [-] (8) edge node[above] {} (9);
    \path [-] (6) edge node[above] {} (11);
    \path [-] (7) edge node[below] {} (6);
    \path [-] (4) edge node[above] {} (3);
    \path [-] (8) edge node[above] {} (10);
    \path [-] (4) edge node[above] {} (10);
    \path [-] (8) edge node[above] {} (11);
\end{tikzpicture}
\,\,\vline\vline
\qquad
\begin{tikzpicture}
    \node[shape=circle,draw=black, fill=black, text=white] (2) at (0,3) {2};
    \node[shape=circle,draw=black, fill=black,text=white] (3) at (2.5,4) {3};
    \node[shape=circle,draw=black, fill=black,text=white] (4) at (5,2) {4};
    \node[shape=circle,draw=black, fill=black,text=white] (5) at (0,1.5) {5};
    \node[shape=circle,draw=black, fill=black,text=white] (6) at (3,-2) {6};
    \node[shape=circle,draw=black, fill=black,text=white] (7) at (4,-2) {7};
    
    \node[shape=circle,draw=black] (3') at (2.5,1) {$3'$};
    \node[shape=circle,draw=black] (2') at (1.5,1.5) {$2'$};
    \node[shape=circle,draw=black] (4') at (3,2) {$4'$};
    \node[shape=circle,draw=black] (1') at (3,-0.5) {$1'$};

    \path [-] (5) edge node[above] {} (9);
    \path [-] (2) edge node[left] {} (5);
    \path [-] (8) edge node[above] {} (9);
    \path [-] (6) edge node[above] {} (11);
    \path [-] (7) edge node[below] {} (6);
    \path [-] (4) edge node[above] {} (3);
    \path [-] (8) edge node[above] {} (10);
    \path [-] (4) edge node[above] {} (10);
    \path [-] (8) edge node[above] {} (11);
\end{tikzpicture}}
        \caption{Left: Starting point of Algorithm \ref{alg:algorithm-final-step}. Right: The tree obtained after the preprocessing step in Algorithm \ref{alg:algorithm-final-step}.}
        \label{fig:tree4}
    \end{figure}

   \begin{algorithm}
    \label{alg:algorithm-final-step}
    The algorithm consists of a preprocessing step followed by steps 1 and 2. Step 2 is further divided into two substeps. Recall that the inputs of the algorithm will be modified reduced components, where each free vertex has at least degree 3 and at most one fixed neighbor vertex.
    
    \begin{itemize}
        \item The pre-processing step gives us an enumeration of the free vertices such that each free vertex has at most one adjacent free vertex with a higher number in the ordering. 
        \item Steps 1 and 2a. rewire the edges of the tree to keep the degree of the free vertices below 2. In case this rewiring produces additional connections between free and fixed vertices, an additional rewiring is performed in Step 2b. See Figure \ref{fig:tree6}.
    \end{itemize}
    
    \textbf{Pre-processing step.}
   Let \(\gamma_0 \in V(\free{F}(\mathfrak{T})) = V(\reduc)\) be the free vertex that has the smallest \(\gamma' \in \fix\partial\mathfrak{T} \subset \N_1\) as its unique fixed neighbor. Starting from \(\gamma_0\), we then enumerate the vertices \(V(\reduc)\) using the breadth--first--search. 
   At each level, prioritize the remaining \(\gamma'\) with the smallest index according to the original vertex labels.\footnote{We fix this prioritization rule for the sake of specificity.}
    This produces an enumeration \(\tilde I \colon [\abs{\reduc}] \to V(\reduc)\). 
    We then reverse this enumeration and define \(I \colon [\abs{\reduc}] \to V(\reduc)\) by setting \({I}(h) \coloneqq \tilde I(\abs{\reduc}-h+1), h =1, \dots, \abs{\reduc}\), see Figure \ref{fig:tree4}.
    The crucial property satisfied by the enumeration is that for each \(\gamma \in V(\reduc)\),
    there exists \emph{at most} one adjacent vertex \(\gamma' \sim \gamma\)
    such that \({I}^{-1}(\gamma')>I^{-1}(\gamma)\). This is a direct consequence of the fact that \(\mathfrak{T}\) is a tree and the labeling \(I\) is based on the BF--search (cf. Theorem 6.2 in \cite{bondy_graph_2008}).

    Let \(n_0 \coloneqq \abs{\reduc}\). For each \(k=1,\dots,n_0\),
    define \(\fix{\gamma}_{k,1}\) to be the unique fixed neighbor of \(I(k)\) whenever such neighbor exists, otherwise set \(\fix{\gamma}_{k,1}\coloneqq \gamma_{\emptyset}\). Note that due to the pre-processing step, we always have \(\fix{\gamma}_{1,1} \neq \gamma_{\emptyset}\).

    Initialize \(\mathfrak{T}_{1}^{\rm r}\) as the weighted \(2\)-color tree
    \[\mathfrak{T}_{1}^{\rm r} \coloneqq (V(\mathfrak{T}_{1}^{\rm r}), E(\mathfrak{T}_{1}^{\rm r}), c_{\mathfrak{T}_{1}^{\rm r}}, w_{\mathfrak{T}_{1}^{\rm r}}),\]
    with
    \begin{align*}
        V({\mathfrak{T}_{1}^{\rm r}}) &\coloneqq V(\mathfrak{T}) \\
        E(\mathfrak{T}_{1}^{\rm r}) &\coloneqq E(\mathfrak{T}) \\
        c_{\mathfrak{T}_{1}^{\rm r}} &\coloneqq c_{j,n}\vert_{V(\mathfrak{T})},
    \end{align*}
    with \(c_{j,n}\) defined in \eqref{eqdef:color-map}, 
    and the weights given by 
    \(w_{\mathfrak{T}_1^{\rm r}} \colon E(\mathfrak{T}_1^{\rm r}) \to [0,+\infty)\) with
    \begin{align*}
        w_{\mathfrak{T}_1^{\rm r}}(\{\gamma,\gamma'\}) \coloneqq 
        \frac{1}{2}\abs{x_{\gamma}-x_{\gamma'}}.
    \end{align*}
    The edge weights encode distance estimates between the corresponding variables. They will be used as a distance input for the Mayer functions in the tree integrals.
    
    The algorithm constructs a finite sequence of weighted trees \((\mathfrak{T}_{k}^{\rm r})_{k=1}^{n_0}\). 
    For each \(k\), let 
      \[N_k \coloneqq \{\gamma' \in V(\mathfrak{T}) : \{{I}(k), \gamma'\} \in E(\mathfrak{T}_{k}^{\rm r})\}\]
    denote the set of neighbours of \( {I}(k)\) in the tree \(\mathfrak{T}_{k}^{\rm r}\). 
    Finally, initialize \(D_1 = \emptyset\) and, for \(k \geq 2\), define \(D_{k} \coloneqq \{I(1),\dots,{I}(k-1)\}\) the set of vertices
    processed before iteration \(k\).
    Throughout the algorithm, at any given iteration \(k\), the set \(N_{k} \setminus D_{k} \subset V(\overline{\reduc})\) is of size \(2\). This is in fact an invariant property of the algorithm.
	
    The preprocessing step is now complete. For each \(k = 1,\dots, n_0-1\), perform Steps 1 and 2 below, see Figure \ref{fig:tree6}. 
	
    \textbf{Step 1.} Let \(\gamma ={I}(k)\). Choose \(\gamma^*_{k} \in N_{k} \setminus D_k\) such that 
			\begin{align}
            \label{ineq:nearest_point}
            \abs{x_{\gamma}-x_{\gamma^*_k}}     \leq \abs{x_\gamma -  x_{\gamma'}}, \quad \text{ for all }\gamma' \in N_{k}\setminus D_k.
            \end{align}
        If both vertices are at the same distance from \(\gamma\), we pick the one with the smaller label.
    
	\textbf{Step 2a.} Let \(\gamma={I}(k)\).
    By construction, \(\gamma\) has at most one free neighbor \(\gamma'=I(k')\) satisfying \( k' >k\).
    We necessarily have either
			\(\gamma^*_k = \gamma'\) or \(\gamma_{k}^* = \fix{\gamma}_{k,k}\). 
    In both cases, 
    the triangle inequality together with \eqref{ineq:nearest_point} yields:
            \[\frac{1}{2}\abs{x_{\gamma'}-x_{\fix{\gamma}_{k,k}}} \leq \abs{x_{\gamma}-x_{\gamma^*_k}}.\]
			We update the tree by replacing the edge \(\{\gamma, \gamma_{k}^\ast\}\) with \(\{\gamma', \fix{\gamma}_{k,k}\}\), that is
            \[E(\mathfrak{T}_{k+1,a}^r) \coloneqq \left(E(\mathfrak{T}_{k}^r)\setminus \{\{\gamma, \gamma_{k}^\ast\}\}\right) \cup \{\{\gamma', \fix{\gamma}_{k,k}\}\}.\] Here, the added edge is associated with a distance factor \(1/2\abs{x_{\gamma'}-x_{\fix{\gamma}_{k,k}}}\). In other words, we set 
            \begin{align*}
                \begin{cases}
                &w_{\mathfrak{T}_{k+1,a}^r}(\{\gamma_1, \gamma_2\}): = w_{\mathfrak{T}^r_k}(\{\gamma_1, \gamma_2\}), \quad \forall \{\gamma_1, \gamma_2\} \in E(\mathfrak{T}_{k}^r)\setminus \{\{\gamma, \gamma_{k}^\ast\}\} \\
                &w_{\mathfrak{T}_{k+1,a}^r}(\{\gamma, \fix{\gamma}_{k,k}\}) \coloneqq \frac{1}{2}w_{\mathfrak{T}^r_k}(\{\gamma, \gamma^\ast_k\}).
                \end{cases}
            \end{align*}
            
            The vertex \(\gamma = {I}(k)\) has now been processed, so we set \(D_{k+1} \coloneqq \{{I}(k)\} \cup D_k\).
            
    \begin{figure}
        \centering
        \resizebox{40em}{8em}{
\begin{tikzpicture}
    \node[shape=circle,draw=black, fill=black, text=white, minimum size=3em] (b1) at (0,3) {$\fix{\gamma}_{k,k}$};
    \node[shape=circle,draw=black, fill=black, text=white,minimum size=3em] (b2) at (5,3) {$\fix{\gamma}_{k',k}$};
    
    \node[shape=circle,draw=black,minimum size=3em] (w1) at (0,0) {$\gamma$};
    \node[shape=circle,draw=black,minimum size=3em] (w2) at (5,0) {$\gamma'$};

    \path [-] (b1) edge node[above] {} (w1);
    \path [-] (w1) edge node[left] {} (w2);
    \path [-] (b2) edge node[above] {} (w2);

    \path [-] (w1) edge (-1.5,0);

    \path[-] (b1) edge (-1,4);
    \path[-] (b2) edge (4,4);

    \path [-] (w2) edge (6.5,0);
    \path [-] (w2) edge (3.5,-1);
\end{tikzpicture}
\qquad
\begin{tikzpicture}
    \node[shape=circle,draw=black, fill=black, text=white, minimum size=3em] (b1) at (0,3) {$\fix{\gamma}_{k,k}$};
    \node[shape=circle,draw=black, fill=black, text=white,minimum size=3em] (b2) at (5,3) {$\fix{\gamma}_{k',k}$};
    
    \node[shape=circle,draw=black,minimum size=3em] (w1) at (0,0) {$\gamma$};
    \node[shape=circle,draw=black,minimum size=3em] (w2) at (5,0) {$\gamma'$};

    \path [-, dotted] (b1) edge node[above] {} (w1);
    \path [-] (b1) edge node[above] {$1/2$} (w2);
    \path [-] (w1) edge node[left] {} (w2);
    
    \path [-] (b2) edge node[above] {} (w2);

    \path [-] (w1) edge (-1.5,0);

    \path[-] (b1) edge (-1,4);
    \path[-] (b2) edge (4,4);

    \path [-] (w2) edge (6.5,0);
    \path [-] (w2) edge (3.5,-1);
    
\end{tikzpicture}
\qquad
\begin{tikzpicture}
    \node[shape=circle,draw=black, fill=black, text=white, minimum size=3em] (b1) at (0,3) {$\fix{\gamma}_{k,k}$};
    \node[shape=circle,draw=black, fill=black, text=white,minimum size=3em] (b2) at (5,3) {$\fix{\gamma}_{k',k}$};
    
    \node[shape=circle,draw=black,minimum size=3em] (w1) at (0,0) {$\gamma$};
    \node[shape=circle,draw=black,minimum size=3em] (w2) at (5,0) {$\gamma'$};

    \path [-, dotted] (b1) edge node[above] {} (w1);
    \path [-, dotted] (b1) edge node[above] {} (w2);
    \path [-] (w1) edge node[left] {} (w2);

    \path [-] (b1) edge node[above] {$1/4$} (b2);
    \path [-] (b2) edge node[above] {} (w2);

    \path [-] (w1) edge (-1.5,0);

    \path[-] (b1) edge (-1,4);
    \path[-] (b2) edge (4,4);

    \path [-] (w2) edge (6.5,0);
    \path [-] (w2) edge (3.5,-1);
    
\end{tikzpicture}
}
        \caption{Possible local operations at an iteration of the algorithm. This is illustrated from the point of view of the vertex labeled with \(\gamma\), with \(\gamma'\) being the unique adjacent vertex such that \(I^{-1}(\gamma')>I^{-1}(\gamma)\).}
        \label{fig:tree6}
    \end{figure}
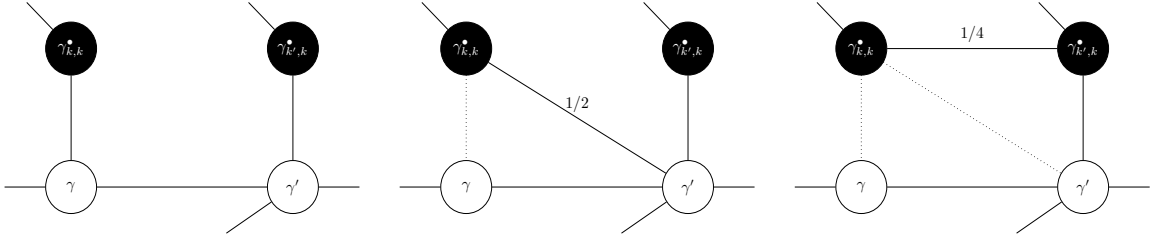
		\textbf{Step 2b.} 
        Consider \(\gamma' = I(k')\). 
        
        If \(\fix{\gamma}_{k', k} = \gamma_{\emptyset}\), set \(\fix{\gamma}_{k', k+1}\coloneqq \fix{\gamma}_{k,k}\). 
            
            Otherwise, choose \(\fix{\gamma}_{k'} \in \{\fix{\gamma}_{k',k}, \fix{\gamma}_{k,k}\}\) such that 
			\begin{align*}
            \abs{x_{\gamma'}-x_{\fix{\gamma}_{k'}}} \leq
					\abs{x_{\gamma'} - x_{\fix{\gamma}}}, \quad \text{ for both } \fix{\gamma} \in \{\fix{\gamma}_{k',k}, \fix{\gamma}_{k,k}\}.
                    \end{align*}
                    
					Let 
					\(\fix{\bar{\gamma}}_{k'}\) be the other element. We have,
					\begin{align*}
                    \frac{1}{2}\abs{x_{\fix{\gamma}_{k',k}}-x_{\fix{\gamma}_{k,k}}} \leq \abs{x_{\gamma'}-x_{\fix{\bar{\gamma}_{k'}}}}.
                    \end{align*}
                    
                    We remove the edge 
					\(\{\gamma', \fix{\bar{\gamma}_{k'}}\}\) and add the edge \(\{\fix{\gamma}_{k',k}, \fix{\gamma}_{k,k}\}\), i.e. we set
                    \[E(\mathfrak{T}^{\rm r}_{k+1}) \coloneqq 
                    \left(
                    E(\mathfrak{T}^{\rm r}_{{k+1},a}) \setminus \{\{\gamma', \fix{\overline{\gamma}_{k'}}\}\}\right) \cup \{\{\fix{\gamma}_{k',k}, \fix{\gamma}_{k,k}\}\},\] 
                    where we associate the distance factor
					 \(\frac{1}{2}\abs{x_{\fix{\gamma}_{k',k}}-x_{\fix{\gamma}_{k,k}}}\) to the latter edge. 
                     In other words, we set 
                     \begin{align*}
                        \begin{cases}
                &w_{\mathfrak{T}^{\rm r}_{k+1}}(\{\gamma_1, \gamma_2\}) \coloneqq w_{\mathfrak{T}^{\rm r}_k}(\{\gamma_1, \gamma_2\}), \quad \forall \{\gamma_1, \gamma_2\} \in E(\mathfrak T^{\rm r}_{k+1,a})\setminus \{\{\gamma', \fix{\gamma}_{k'}\}\} \\
                &w_{\mathfrak{T}^{\rm r}_{k+1}}(\{\gamma, \fix{\gamma}_{k,k}\}) \coloneqq \frac{1}{2}w_{\mathfrak T^{\rm r}_{k+1,a}}(\{\fix{\gamma}_{k', k}, \fix{\gamma}_{k,k}\}).
                \end{cases}
            \end{align*}
					Set \(\fix{\gamma}_{k', k+1} \coloneqq \fix{\bar{\gamma}}_{k'}\) for \(k'\) and for all other \(k'' > k\), set \(\fix{\gamma}_{k'', k+1} = \fix{\gamma}_{k'', k}\). 
                    
                    Repeat Steps 1 and 2 for \(k=1,\dots, n_0-1\). After the final iteration, the algorithm returns the weighted tree \(\mathfrak{T}^{\mathrm{r}} =\mathfrak{T}^{\mathrm{r}}_{n_0}\).
        \end{algorithm}
		
        The output tree \(\mathfrak{T}^r\) satisfies the following properties:
        \begin{enumerate}[label=(\roman*)]
            \item Each free vertex \(\gamma\) is connected to at most two other vertices. Moreover, for each  \(k = 1, \dots, n_0 = \abs{\reduc}\), either \({I}(k)\) has a neighbor \(\gamma' = {I}(k')\) with \(k' > k\) or \({I}(k)\) has a fixed neighbor \(\fix{\gamma} \in \fix{\partial}\reduc\). 
            \item  The edges involving fixed vertices might have modified weights, 
        but in total this can be controlled, since at the final step
        \begin{align*}
        w_{\finaltr}(\{\gamma, \gamma'\}) \geq \frac{1}{4}\abs{x_{\gamma}-x_{\gamma'}}, \quad \{\gamma, \gamma'\} \in E(\finaltr).
        \end{align*}
        \item Since the produced tree depends on \(\vc{x}_{V(\reduc)} \in D\), Steps 1 and 2 can lead to bifurcations. The combinatorial effect of having to consider all possible bifurcation points in Steps 2a) and 2b) leads to an exponential bound \[2^{n_0-1} 2^{n_0-1} = 4^{n_0 - 1} = 4^{\abs{\reduc}-1}.\]
        \end{enumerate}

        Finally, we estimate 
        \begin{align*}      
        \inI{T\left(\overline{\reduc};D\right)}^D \leq \sum_{\mathfrak{T}^{\rm r} \in \mathcal{T}_{j, n_0}} \cf{\mathfrak{T}^{\rm r} \text{ accessible from }   {T\left(\overline{\reduc};D\right)}} \inI{\mathfrak T^{\rm r}}(\vc{x}_{\fix{\partial}C_{i}}).
        \end{align*}
        
       Each tree--integral \(\inI{\finaltr}(\vc{x}_{\fix \partial C_i})\) is estimated by integrating out the free variables one after the other using Fubini's theorem. Because of property (i), this can be obtained using \eqref{eq:int_l} and \eqref{eq:int_c}. Iterating this procedure over \(k = 1, \dots, n_0 = \abs{\reduc}\), we always have a free variable \(x_{\gamma}, \gamma \in [n+j]\setminus [j]\) that can be integrated out. 

        Hence, we obtain the bound
        \begin{align*}
            \inI{\mathfrak T^{\rm r}}(\vc{x}_{\fix{\partial}C_{i}}) \leq (4\, \bar \rho \, C_{\beta,s}\,\omega(\varepsilon))^{\abs{\reduc}}\left( \prod_{\{\gamma, \gamma'\} \in E(\fix F(\mathfrak T^{\rm r}))}\abs{\eta^{\varepsilon,\alpha}_{1/4}(x_{\gamma}, x_{\gamma'})}\right),
        \end{align*}
        where the rescaled Mayer--function \(\eta_{1/4}^{\varepsilon, \alpha}\) is defined in \eqref{defeq:eta_dil}.
        
        This concludes the proof of the following Lemma \ref{lemma:main-upper-bound}.

    \begin{lemma}
    \label{lemma:main-upper-bound}
    Given \(j \in \N_1\), \(n \in \N_0\), \(\vc x_j \in \Lambda^j\), and a tree \(T \in \mathcal{T}_{j,n}\), 
    it holds that
    \begin{multline*}
        \inI{T}^{\varepsilon,\alpha}(\vc{x}_j) \leq \left(\prod_{\{\gamma, \gamma'\} \in E(\fix{F}(T))} |\eta^{\varepsilon,\alpha}(x_{\gamma}, x_{\gamma'})|\right)\\
        \times\prod_{(C_{i})_{i=1}^{\ell_1} \in \mathcal{C}(\free{F}(T))} \sum_{D \in \mathcal{D}_i} (C_{\beta, s} \bar \rho\, \omega(\varepsilon))^{\mathfrak{l}(C_i)+\mathfrak{c}(C'_i)}\inI{T(\overline{\reduc};D)}^D(\vc{x}_{\fix{\partial}C_i}),
    \end{multline*}
        where \(\abs{\mathcal{D}_i} \leq \exp\{\frac{\abs{\fix{\partial}C_i}}{{\rme}}\}\) and, for each tree \(\mathfrak{T}=T(\overline{\reduc}; D)\), the following bound holds
        \begin{align*}
            \inI{\mathfrak{T}}^{D}(\vc{x}_{\fix{\partial}C_{i}}) \leq (4\bar \rho\, C_{\beta,s} \omega(\varepsilon))^{\abs{\reduc}}\max_{\finaltr \succ  \mathfrak{T}}\left( \prod_{\{\gamma,\gamma'\} \in E(\fix F(\mathfrak{T}^{\rm r}))}\abs{\eta^{\varepsilon,\alpha}_{1/4}(x_{\gamma}, x_{\gamma'})}\right).
        \end{align*}
        Here we recall that the notation \(\finaltr \succ \mathfrak{T}\) means that the tree \(\finaltr\) is \emph{accessible} from \(\mathfrak{T}\), or in other words is one of the outputs of the algorithm for some \(\vc{x}_{V(\reduc)} \in D_i\).
    \end{lemma}

        \subsection{Conclusion}
        \label{sec4.6}
    
        We conclude by showing how to obtain the thesis of Proposition \ref{graph_red}.
By Lemma \ref{lemma:main-upper-bound}, we obtain
    \begin{multline}
                \label{ineq:main-long-range-upper-bound}
        \mathfrak I_{T}^{\varepsilon,\alpha}(\vc x_j)
        \leq \Biggl[\Biggl(\prod_{\{\gamma, \gamma'\} \in E(\fix{F}(T))} \abs{\eta^{\varepsilon,\alpha}(x_{\gamma}, x_{\gamma'})}\Biggr) \\
        \qquad \times \prod_{(C_{i})_{i=1}^{\ell_1} \in \mathcal{C}(\free{F}(T))}\sum_{D \in \mathcal{D}_i} (C_{\beta, s} \bar \rho \,\omega(\varepsilon))^{\mathfrak{l}(C_i)+\mathfrak{c}(C'_i)} \nonumber \\
        \qquad\times (4\, \bar \rho\, C_{\beta,s}\, \omega(\varepsilon))^{\abs{\reduc}}\max_{\finaltr \succ \mathfrak{T}}\left( \prod_{\{\gamma,\gamma'\} \in E(\fix F(\finaltr))}\abs{\eta^{\varepsilon,\alpha}_{1/4}(x_{\gamma}, x_{\gamma'})}\right)\Biggr]. \nonumber 
    \end{multline}
    Since \(\abs{\mathcal{D}_i} = \abs{\mathcal{D}(C_i; \vc{x}_{\fix{\partial}C_i})} \leq \rme^{\abs{\fix{\partial}C_i}/\rme}\)
    by Lemma \ref{lemma_mod_conn}, we can apply a \((1,+\infty)\)-Hölder's inequality to the sum over the trees in this family to get
    \begin{multline}
        \mathfrak I_{T}^{\varepsilon,\alpha}(\vc x_j)
        \leq  \left(\prod_{\{\gamma, \gamma'\} \in E(\fix F({T}))} \abs{\eta^{\varepsilon,\alpha}(x_{\gamma}, x_{\gamma'})}\right) \prod_{(C_{i})_{i=1}^{\ell_1} \in \mathcal{C}(\free F({T}))}\rme^{\abs{\fix{\partial}C_i}/\rme}
         (4\,C_{\beta, s} \bar\rho \,\omega(\varepsilon))^{\abs{C_i}} \nonumber \\
         \qquad \times \max_{\finaltr \succ \mathfrak{T}}\left( \prod_{\{\gamma,\gamma'\} \in E(\fix F(\finaltr))}\abs{\eta^{\varepsilon,\alpha}(x_{\gamma}, x_{\gamma'})}\right). \nonumber 
     \end{multline}
     Notice that, by the definition of connected components and Lemma \ref{lemma:component-number-bound}, 
     \[
     \sum_{i=1}^{\ell_1}\abs{C_i}=n \qquad \text{and} \qquad \sum_{i=1}^{\ell_1} \abs{\fix{\partial}C_i} \le \ell_1+j-1 \le 2(j-1).
     \]
     Since for any sequence of trees  \((\finaltr_{i})_{1=1}^{\ell_1}\), each depending on \(\fix \partial C_i\) respectively, the graph 
     \[\left([j], E(\fix{F}(T)) \cup \left( \bigcup_{i=1}^{\ell_1} E\left(\fix F\left(\finaltr_i\right)\right)\right)\right)\] is a tree on \(j\) vertices, we obtain the following upper bound which is valid for all \(\vc{z}_j \in (\Lambda \times \R^d)^{j}\) and all \(T \in \mathcal{T}_{j,n}\):
    \begin{align*}
        \label{ineq:upper-bound-trucnated-correlation}
        \mathfrak I_{T}^{\varepsilon,\alpha}(\vc x_j)
        &\leq \rme^{2(j-1)/\rme}(4\, C_{\beta,s} \, \bar \rho \, \omega(\varepsilon))^{n}\max_{T' \in \mathcal{T}_{j,0}} \left(\prod_{\{\gamma, \gamma'\} \in E(T')}\abs{\eta^{\varepsilon,\alpha}(x_{\gamma}, x_{\gamma'})}\right).
     \end{align*}
     \qed
    
   \subsection*{Acknowledgments}
    Co-funded by the European Union (ERC CoG KiLiM, project number 101125162).\footnote{Views and opinions expressed are however those of the authors
only and do not necessarily reflect those of the European Union or the European Research
Council. Neither the European Union nor the granting authority can be held responsible for
them.}

\bibliographystyle{plain}
\bibliography{bibliography}
\end{document}